\documentclass[a4paper,USenglish]{lipics-v2019}
\nolinenumbers
\hideLIPIcs
\makeatletter
\renewcommand\@oddfoot{
	\hfil
	\rlap{%
		\vtop{%
			\vskip10mm
			\colorbox[rgb]{0.99,0.78,0.07}
			{\@tempdima\evensidemargin
				\advance\@tempdima1in
				\advance\@tempdima\hoffset
				\hb@xt@\@tempdima{%
					\textcolor{darkgray}{\normalsize\sffamily
						\bfseries\quad
						\expandafter\textsolittle\expandafter{
							arXiv.org}}%
					\strut\hss}}}}
}

\theoremstyle{plain}

\usepackage{refcount}

\usepackage{tikz}
\usetikzlibrary{arrows.meta,decorations.pathreplacing,positioning,calc}
\tikzset{Tip/.tip={Triangle[angle=30:5pt]}}
\title{Efficient Simulation of Combinational Circuits in MapReduce}
\titlerunning{Efficient Simulation of Combinational Circuits in MapReduce}

\title{Efficient Implementation of Combinational Circuits in MapReduce}
\titlerunning{Efficient Implementation of Combinational Circuits in MapReduce}

\title{Efficient Circuit Simulation in MapReduce}
\titlerunning{Efficient Circuit Simulation in MapReduce}

\author{Fabian Frei}{Department of Computer Science, ETH Zürich, Universitätstrasse 6, CH-8006 Zürich, Switzerland}{fabian.frei@inf.ethz.ch}{}{}
\author{Koichi Wada}{Department of Applied Informatics, Hosei University, 3-7-2 Kajino, 184-8584 Tokyo, Japan}{wada@hosei.ac.jp}{}{Research done in part during a supported visit at ETH Zürich and partly supported by JSPS KAKENHI No. 17K00019 and by the Japan Science and Technology Agency (JST) SICORP (Grant\#JPMJSC1806).}

\makeatletter
\g@addto@macro\bfseries{\boldmath}
\makeatother

\newcommand{\INq}{\mathcal{X}_q}
\newcommand{\partialassignmentq}{\assignment_q}
\newcommand{\level}{\textnormal{level}}
\newcommand{\depth}{\textnormal{depth}}
\newcommand{\val}{\textnormal{val}}
\newcommand{\size}{\textnormal{size}}
\newcommand{\up}{\textnormal{up}}
\newcommand{\down}{\textnormal{down}}
\newcommand{\PBP}{\textnormal{PBP}}
\newcommand{\localy}{y^{\textnormal{local}}}
\newcommand{\assignment}{\alpha}
\newcommand{\formeralpha}{\gamma}
\newcommand{\id}{\textnormal{id}}
\newcommand{\To}{\textnormal{to}}
\newcommand{\From}{\textnormal{from}}
\newcommand{\In}{\textnormal{in}}
\newcommand{\Out}{\textnormal{out}}

\newcommand{\subprogramNumber}{\#S} %

\newcommand{\subcircuitDepth}{s}

\renewcommand{\epsilon}{\varepsilon}

\newcommand{\SPACE}{\mathcal{SP}\hspace{-.21em}\mathcal{A}\hspace{.08em}\mathcal{C}\hspace{-.02em}\mathcal{E}}
\newcommand{\MRC}{\mathcal{M}\hspace{-.02em}\mathcal{R}\hspace{.12em}\mathcal{C}}
\newcommand{\DMRC}{\mathcal{D}\hspace{-.05em}\mathcal{M}\hspace{-.02em}\mathcal{R}\hspace{.1em}\mathcal{C}}
\newcommand{\NC}{\mathcal{N}\hspace{-.07em}\mathcal{C}}
\newcommand{\AC}{\mathcal{A}\hspace{.1em}\mathcal{C}}
\newcommand{\EREW}{\mathcal{E}\hspace{-.02em}\mathcal{R}\hspace{.07em}\mathcal{E}\hspace{-.02em}\mathcal{W}}
\newcommand{\CREW}{\mathcal{C}\hspace{-.00em}\mathcal{R}\hspace{.07em}\mathcal{E}\hspace{-.02em}\mathcal{W}}
\newcommand{\CRCW}{\mathcal{C}\hspace{-.00em}\mathcal{R}\hspace{.07em}\mathcal{C}\hspace{-.02em}\mathcal{W}}

\newcommand{\basis}{\mathcal{B}}
\newcommand{\mem}{m}
\newcommand{\dum}{\textnormal{dummy}}
\newcommand{\psum}{\textnormal{p-sum}}
\newcommand{\last}{\textnormal{last}}
\newcommand{\N}{\mathbb{N}}
\newcommand{\range}[1]{[#1]}
\newcommand{\BranchingProgram}{B}
\newcommand{\processor}{\textrm{P}\!}
\newcommand{\register}{\textrm{R}}


\authorrunning{F.\,Frei, K.\,Wada}

\Copyright{Fabian Frei, Koichi Wada}

\relatedversion{This is the full version of  preliminary paper with the same title~\cite{FW} presented at the 30th International Symposium on Algorithms and Computation (ISAAC 2019).}


\acknowledgements{We thank the anonymous reviewers for their helpful comments.}

\begin{document}
	
	\maketitle
	
	\begin{abstract}
		The MapReduce framework has firmly established itself as one of the most widely used parallel computing platforms for processing big data on tera- and peta-byte scale. 
		Approaching it from a theoretical standpoint has proved to be notoriously difficult, however. 
		In continuation of Goodrich et al.'s early efforts, explicitly espousing the goal of putting the MapReduce framework on footing equal to that of long-established models such as the PRAM, 
		we investigate the obvious complexity question of how the computational power of MapReduce algorithms compares to that of combinational Boolean circuits commonly used for parallel computations. 
		Relying on the standard MapReduce model introduced by Karloff et al.\ 
		a decade ago, we develop an intricate simulation technique to show 
		that any problem in $\mathcal{NC}$ (i.e., a problem solved by a logspace-uniform family of Boolean circuits of polynomial size and a depth polylogarithmic in the input size) 
		can be solved by a MapReduce computation in $O(T(n)/\log n)$ rounds, 
		where $n$ is the input size and $T(n)$ is the depth of the witnessing circuit family. 
		Thus, we are able to closely relate the standard, uniform $\NC$ hierarchy modeling parallel computations to the deterministic MapReduce hierarchy $\DMRC$ by proving that $\NC^{i+1}\subseteq\DMRC^i$ for all $i\in\N$. 
		Besides the theoretical significance, this result has important applied aspects as well. In particular, we show for all problems in $\NC^1$---many practically relevant ones such as integer multiplication and division, the parity function, and recognizing balanced strings of parentheses being among these---how to solve them in a constant number of deterministic MapReduce rounds. 
	\end{abstract}

	\section{Introduction}
	Despite the overwhelming success of the MapReduce framework in the big data industry and the great attention it has garnered ever since its inception over a decade ago, theoretical results about it have remained scarce in the literature. 
	In particular, it is very natural to ask how powerful exactly MapReduce computations are in comparison to the traditional models of parallel computations based on circuits, a question that has practical implications as well. 
	The answers have proved to be very elusive, however. 
	In this paper, we show how MapReduce programs can efficiently simulate circuits used for parallel computations, thus tying these two worlds together more tightly. 
	
	In this section we first provide an introduction to the concept of MapReduce, then present the related work, and finally describe our contribution. 
	In Section~\ref{Sec2}, we will formally define the traditional models of parallel computing and the MapReduce model. In Section~\ref{Sec3}, we then derive our main results. 
	Section~\ref{Sec4} concludes the paper with a short summary and a discussion of our findings, outlining opportunities for future research. 
	
	\subsection{Background and Motivation}
	In recent years the amount of data available and demanding analysis has grown at an astonishing rate. The amount of memory in commercially available servers has also grown at a remarkable pace in the past decade; it is now exceeding tera- and even peta-bytes.  Despite the considerable advances in the availability of computational power, traditional approaches remain insufficient to cope with such huge amounts of data. A new form of parallel computing has become necessary to deal with these enormous quantities of available data. 
	The MapReduce framework has been attracting great interest due to its suitability for processing massive data-sets. This framework was originally developed by Google~\cite{DG}, but an open source implementation called Hadoop has recently been developed and is currently used by over a hundred companies, including Yahoo!, Facebook, Adobe, and IBM~\cite{W}. 
	
	MapReduce differs substantially from previous models of parallel computation in that it combines aspects of both parallel and sequential computation. Informally, a MapReduce computation can be described as follows.
	
	The input is a set of \emph{key-value pairs} $\langle k;v \rangle{}$. 
	In a first step, the {\em map step}, each of these key-value pairs is separately and independently transformed into an entire set of key-value pairs by a {\em map function} $\mu$. 
	In the next step, the {\em shuffle step}, we collect all key-value pairs from the sets that have been produced in the previous step, group them by their keys, and merge each group $\{\langle k;v_1 \rangle{}, \langle k;v_2 \rangle{}, \ldots\}$ of pairs containing the same key into a single key-value pair $\langle k;\{v_1,v_2,\ldots\} \rangle{}$ consisting of said key and a list of the associated values. 
	In a third step, the {\em reduce step}, a {\em reduce function} $\rho$ transforms the list of values in each key-value pair $\langle k;\{v_1,v_2,\ldots\} \rangle{}$ into a new list $\{v_1',v_2',\dots\}$. Again, this is done separately and independently for each pair. The final output consists of the pairs $\{\langle k;v_1' \rangle{}, \langle k;v_2' \rangle{}, \ldots\}$ for each key $k$. The different instances that implement the reduce function for the different groups of pairs are called reducers. Analogously, mappers are instances of the map function. 
	
	The three steps described above constitute one {\em round} of the MapReduce computation and transform the input set into a new set of key-value pairs. 
	A complete MapReduce computation consists of any given number of rounds and acts just as the composition of the single rounds. The shuffle step works the same way every time; the map and reduce functions, however, may change from round to round. 
	A MapReduce computation with $R$ rounds is therefore completely described by a list $\mu_1,\rho_1,\mu_2,\rho_2,\ldots,\mu_R,\rho_R$ of map and reduce functions. 
	In both the map step and the reduce step, the input pairs can be processed in parallel since the map and reduce functions act independently on the pairs and groups of pairs, respectively. 
	These steps therefore capture the parallel aspect of a MapReduce computation, whereas the shuffle step enforces a partial sequentiality since the shuffled pairs can be output only once the previous map step is completed in its entirety. 
	
	The MapReduce paradigm has been introduced in \cite{DG} in the context of algorithm design and analysis. A treatment as a formal computational model, however, was missing in the beginning. 
	Later on, a number of models have emerged to deal more rigorously with algorithmic issues \cite{FMSSS,GSZ,KSV,P,PPRSU}. %
	In this paper, our interest lies in studying the MapReduce framework from a standpoint of parallel algorithmic power by comparing it to standard models of parallel computation such as Boolean circuits and parallel random access machines (PRAMs). A PRAM can be classified by how far simultaneous access by processors to its memory is restricted; it can be CRCW, EREW, CREW, or ERCW, where R, W, C, and E stand for Read, Write, Concurrent, and Exclusive, respectively \cite{CLR}. 
	If concurrent writing is allowed, we need to further specify how parallel writes by multiple processors to a single memory cell are handled. 
	The most natural choice is arguably that every memory cell contains after each time step the total of all numbers assigned to it by different processors during that step. In fact, all constructions in this paper work with this treatment of simultaneous writes; we thus generally assume this model. If the context warrants it, we speak of a Sum-CRCW 
	to make this assumption explicit. 
	
	\subsection{Related Work}
	We briefly present and discuss the following known results on the comparative power of the MapReduce framework and PRAM models.  
	
	\begin{enumerate}
		\item A $T$-time EREW-PRAM algorithm can be simulated by an $O(T)$-round MapReduce algorithm, where each reducer uses memory of constant size and an aggregate memory proportional to the amount of shared memory required by the PRAM algorithm~\cite{GSZ,KSV}.
		\item A $P$-processor, $M$-memory, $T$-time EREW-PRAM algorithm can be simulated by an $O(T)$-round, $(P+M)$-key MUD algorithm with a communication complexity of $O(\log(P+M))$ bits per key, where a MUD (massive, unordered, distributed) algorithm is a data-streaming MapReduce algorithm in the following sense: The reducers do not receive the entire list of values associated with a given key at once, but rather as a stream to be processed in one pass, using only a small working memory determining the communication complexity~\cite{FMSSS}.
		\item When using MapReduce computations to simulate a CRCW-PRAM instead, again with $P$ processors and $M$ memory, we incur an $O(\log_{m}(P+M))$ slowdown compared to the simulations above, where $m$ is an upper bound on each reducer's input and output~\cite{GSZ}. 
	\end{enumerate}
	
	These results imply that any problem solved by a PRAM with a polynomial number of processors and in polylogarithmic %
	time $T$ can be simulated by a MapReduce computation with 
	an amount of memory equal to the number of PRAM processors, and in a number of rounds equal to the computation time of even the powerful CRCW-PRAM.
	Since the class of problems solved by CRCW-PRAMs in time $T\in O(\log^{i}n)$ is equal to the class of problems solved by 
	families of polynomial-sized combinational circuits consisting of gates with unbounded fan-in and fan-out and time $T\in O(\log^{i}n)$ (often denoted $\AC^i$)~\cite{AB}, 
	these circuits can be simulated in a MapReduce computation with a number of rounds equal to the time required by these circuits.
	
	Since the publication of the seminal paper by Karloff et al.~\cite{KSV}, extensive effort has been spent on developing efficient algorithms in MapReduce-like frameworks~\cite{CKLYBNO,FEGK,KMV,KR,SASV}. Only few relationships between the theoretical MapReduce model~\cite{KSV} and classical complexity classes have been established, however; 
	for example, any problem in $\SPACE(o(\log n))$ 
	can be solved by a MapReduce computation with a constant number of rounds~\cite{FKLRT}.
	
	Recently, Roughgarden et al.~\cite[Theorem 6.1]{RVW} described a short and simple way of simulating $\NC^1$ circuits with a certain class of models of parallel computation. 
	The constraints of these models, namely the number of machines and the memory restrictions, are exactly tailored to allow for this general simulation method, however. In particular, it crucially relies on the fact that all models of this class are more powerful than the MapReduce model in that they all grant us a number of machines that is polynomial in the input size; this makes it possible to just dedicate one machine to each of the circuit gates. Such a simple simulation is impossible with MapReduce computations since the standard model due to Karloff only allows for a sublinear number of machines with sublinear memory. 
	
	\subsection{Contribution}
	We prove that
	$\NC^{i+1} \subseteq \DMRC^i$ for all $i\in\{0,1,2,\dots\}$, where $\DMRC^{i}$ is the set of problems solvable by a deterministic MapReduce computation in $O(\log^in)$  rounds. %
	In the case of $\NC^1 \subseteq \DMRC^0$, which already opens up a plethora of applications on its own, the result holds for every possible choice of $\epsilon$, that is, for $0<\epsilon\le 1/2$. The higher levels of the hierarchy require an entirely different proof method, which yields the result for 
	$0<\epsilon<1/2$. 
	
	This is a substantial improvement over the previous results that only imply, as outlined above, the far weaker claim $\AC^i \subseteq \MRC^i$. The case $i=1$ is of particular practical interest since $\NC^{1}\setminus\AC^{0}$ contains plenty of relevant problems such as integer multiplication and division, the parity function, and the recognition of the Dyck languages $D_n$, which contain all balanced strings of $n$ different types of parentheses; see~\cite{AB}. Our results show how to solve all of these problems with a deterministic MapReduce program in a constant number of rounds. 
	
	\section{Preliminaries}\label{Sec2}
	We denote by $\N=\{0,1,2,\ldots\}$ the natural numbers including zero and let $\N_+=\N \setminus \{0\}$. Moreover, 
	we let $\range{i}=\{0,1,\ldots,i-1\}$ denote the $i$ first natural numbers for any $i\in \N_+$. 
	
	\subsection{Models of Parallel Computation}
	In this section, we define the common complexity classes capturing the power of parallel computation; most prominently the $\NC$ hierarchy. %
	
	A finite set $\basis=\{f_0,\ldots,f_{|\basis|-1}\}$ of Boolean functions 
	$f_i: \{0,1\}^{n_i} \rightarrow \{0,1 \}$ with $n_i\in\N$ for every $i\in[|\basis|]$ is called a {\em basis}. 
	For every $n,m \in \N_+$, a {\em (Boolean) circuit} $C$ over the basis $\basis$ with $n$ inputs
	and $m$ outputs is a directed acyclic graph that contains $n$ \emph{sources} (nodes with no incoming edges),
	called the {\em input nodes}, and $m$ \emph{sinks} (nodes with no outgoing edges). 
	The \emph{fan-in} of a node is the number of incoming edges, the \emph{fan-out} is the number of outgoing edges. 
	Nodes that are neither sources nor sinks are called \emph{gates}. Each gate is labeled with a function $f_i\in \basis$ and has fan-in $n_i$. It computes $f_i$ on the input given by the incoming edges and outputs the result (either 0 or 1) to the outgoing edges. 
	A basis $\basis$ is said to be {\em complete} if for every Boolean function $f$, we can construct a circuit of the described form that computes $f$ over the basis $\basis$. In the following, we use the complete basis $\basis=\{ \vee, \wedge, \neg \}$. %
	
	The {\em size} of a circuit $C$, denoted by $\size(C)$, is the total number of edges it contains. 
	The {\em level} of a node $v$ in a circuit $C$, denoted $\level(v)$, is defined recursively: The level of a sink is $0$, and the level of a node $v$ with nonzero fan-out is one greater than the maximum of the levels of the outgoing neighbors of $v$. 
	The {\em depth} of $C$, denoted $\depth(C)$, is the maximum level across all nodes in $C$. 
	
	A function $f : \{0,1\}^{*} \rightarrow \{0,1\}^{*}$ is {\em implicitly logspace computable} if the two mappings $(x,i)\mapsto \chi_{i\le |f(x)|}$, where $\chi$ denotes the characteristic function, and $(x, i) \mapsto (f(x))_i$ are computable using logarithmic space.
	A circuit family $\{C_n\}_{n=0}^\infty$ is {\em logspace-uniform} if there is an implicitly logspace computable function mapping $1^n$ to the description of the circuit $C_n$. 
	It is known that the class of languages that have logspace-uniform circuits of polynomial size equals $\mathcal{P}$ \cite[Thm. 6.15]{AB}.
	
	For any $i\in\N$, the complexity class $\NC^i$ contains a language $L$ exactly if there is a constant $c$ and a logspace-uniform family of circuits $\{C_n\}_{n=0}^\infty$ recognizing $L$ such that $C_n$ has size $O(n^c)$, depth $O(\log^in)$, and all nodes have fan-in at most 2. 
	The union is Nick's class $\NC=\bigcup_{i=0}^\infty\NC^i$. 
	We mention that there is an analogous definition of classes $\text{Nonuniform-}\NC^i$ that do not require logspace uniformity from the circuits; they constitute a different hierarchy. %
	
	The complexity classes $\AC^i$ and $\AC=\bigcup_{i=0}^\infty\AC^i$ are defined exactly as $\NC^i$ and $\NC$, except that the restriction of the maximal fan-in to at most 2 is omitted. Nevertheless, the restriction on the circuit size imply that the fan-in of a node is bounded by a polynomial in $n$. The OR gates and AND gates in such a circuit can therefore be replaced by trees of gates of fan-in at most 2 with a depth in $O(\log n)$. It follows that $\AC^i \subseteq \NC^{i+1}$ for all $i\in\N$ and thus $\NC=\AC$. (Analogously, we see why Nick's class can also be defined, as it often is, by upper-bounding the fan-in by an arbitrary constant greater than 2.) The inclusion $\NC^i\subseteq\AC^i$ for every $i\in\N$ is immediate from the definition. 
	The first two inclusions of the resulting chain are known to be strict---namely, we have $\NC^0 \subsetneq \AC^0 \subsetneq \NC^1$; see \cite{AB}. 
	
	Finally, we summarize the known results on how the classes of languages recognized by different PRAMs fit into the two hierarchies of $\NC$ and $\AC$. 
	Let $\EREW^i$, $\CREW^i$ and $\CRCW^i$ denote the sets of problems of size $n$ computed by EREW-PRAMs, CREW-PRAM, and CRCW-PRAMs, respectively, with a polynomial number of processors in $O(\log^in)$ time. 
	For every $i\in\N$, we have 
	$\NC^i \subseteq \EREW^i \subseteq \CREW^i \subseteq \CRCW^i = \AC^i \subseteq \NC^{i+1}$; see \cite{AB}.
	
	\subsection{The MapReduce Model}
	In this section we describe the standard MapReduce model as proposed Karloff et al.~\cite{KSV}. 
	It defines the notions of {\em map functions} and {\em reduce functions}, which are summarized under the term {\em primitives}. Roughly speaking, a MapReduce computing system executes primitives, interleaved with so-called \emph{shuffle} operations. The basic data unit in these computations is an ordered pair $\langle key;value\rangle$, called \emph{key-value pair}. 
	In general, keys and values are just binary strings, allowing us to encode all the usual entities. 
	
	A map function is a (possibly randomized) function that takes as input a single key-value pair and outputs a finite multiset of new key-value pairs. 
	A reduce function (again, possibly randomized) takes instead an entire set of key-value pairs $\{\langle k;v_{k,1}\rangle, \langle k;v_{k,2}\rangle , \ldots\}$, where all the keys are identical, and outputs a single key-value pair $\langle k;v'\rangle$ with that same key. 
	
	A MapReduce program is nothing else than a sequence 
	$\mu_1, \rho_1, \mu_2, \rho_2, \ldots, \mu_R, \rho_R$ of map functions $\mu_r$ and
	reduce functions $\rho_r$. The input of this program is a multiset $U_0$ of key-value pairs. 
	For each $r\in \{1,\ldots,R\}$, a map step, a shuffle step and a reduce step are successively executed as follows:
	
	\begin{enumerate}
		\item {\bf Map step:} Each pair $\langle k;v\rangle$ in $U_{r-1}$ is given as input to an arbitrary instance of the map function $\mu_r$, which then produces a finite sequence  %
		of pairs. The multiset of all produced pairs is denoted by $U'_r$. 
		\item {\bf Shuffle step:} For each key $k$, let $V_{k,r}$ be the multiset of all values $v_i$ such that $\langle k,v_i\rangle{}$. The MapReduce system automatically constructs the multiset $V_{k,r}$ from $U'_r$ in the background. 
		\item {\bf Reduce step:} For each key $k$, a reducer (i.e.,  an instance calculating the reduce function $\rho_r$)  receives $k$ and the elements of $V_{k,r}$ in arbitrary order. We usually write such an input as a set of key-value pairs that all have key $k$. The reducer calculates from $V_{k,r}$ another multiset $V_{k,r}'$, which is output in the form of one key-value pair $\langle k,v'\rangle$ for  each $v'$ in $V_{k,r}'$.
	\end{enumerate}
	
	Fix any $\epsilon$ with $0<\epsilon\le 1/2$ and denote the size of the MapReduce program's input by $N$. For every $i\in\N$, a problem is in $\MRC^{i}$ if and only if if there is a MapReduce program $\mu_1, \rho_1, \mu_2, \rho_2, \ldots, \mu_R, \rho_R$ satisfying the following properties: 
	\begin{enumerate}
		\item It outputs a correct answer to the problem with probability at least $3/4$. 
		\item The number of rounds of the MapReduce program, $R$, is in $O(\log^{i}N)$. 
		\item The potentially randomized primitives (i.e., all map and reduce functions) are computable by a RAM with $O(\log N)$-bit words using $O(N^{1-\epsilon})$ space and time polynomial in $N$.  
		\item 
		The pairs produced by the map functions can be stored in $O(N^{2(1-\epsilon)})$ space. 
	\end{enumerate}
	
	A MapReduce program satisfying these conditions is called an $\MRC^{i}$-algorithm. 
	Note that due to the last condition it is impossible to even store the input unless $2(1-\epsilon) \geq 1$, which explains the restriction to $0<\epsilon\le 1/2$. 
	As with $\NC$, we define the union class $\MRC=\bigcup_{i=0}^\infty\MRC^i$. 
	Requiring all primitives to be deterministic yields the analogous hierarchy of $\DMRC=\bigcup_{i=0}^\infty\DMRC^i$. 
	Note that we obviously have $\DMRC^i\subseteq\MRC^i$ for all $i\in\N$. 
	We will often refer to the single rounds of such MapReduce algorithms as $\MRC$-rounds and $\DMRC$-rounds, respectively. 
	
	\section{Simulating Parallel Computations by MapReduce}\label{Sec3}%
	We are now going to prove our two main results $\NC^1\subseteq\DMRC^0$ for $0<\epsilon \leq 1/2$ and $\NC^{i+1}\subseteq\DMRC^i$ for all $i\in\N_+$ and $0<\epsilon<1/2$ in Sections~\ref{Sec31} and \ref{Sec32}, respectively. 
	In both cases, we will be making use of the technical tool derived in Section~\ref{technicaltools} and obtain the results by showing how to use MapReduce computations for two different, delicate simulations. 
	For the inclusion $\NC^1\subseteq\DMRC^0$, we simulate width-bounded branching programs that are equivalent to the respective circuits by Barrington's classical theorem~\cite{B}, whereas for the higher levels of the hierarchy, we directly simulate the combinational circuits themselves.
	
	\subsection{A Technical Tool}\label{technicaltools}
	Goodrich et al.~\cite{GSZ} parametrize MapReduce algorithms, on the one hand, by the memory limit $m$ for the input/output buffer of the reducers and, on the other hand, by the {\em communication complexity} $K_r$ of round $r$, that is, the total size of inputs and outputs for all mappers and reducers in round $r$. We state a useful result from~\cite{GSZ}.
	
	\begin{theorem}\label{CRCW-GSZ}
		Any CRCW-PRAM algorithm using $M$ total memory, $P$ processors and $T$ time can be simulated in $O(T\log_{\mem}P)$
		deterministic MapReduce-rounds with communication complexity $K_r\in O((M+P)\log_{\mem}(M+P))$.%
	\end{theorem}
	
	We denote by $N$ the size of the smallest circuit representation of the CRCW-PRAM algorithm (i.e., its number of edges) plus the size of its input. 
	Taking into account our requirements $\mem\in O(N^{1-\epsilon})$ and $K_r\in O(N^{2(1-\epsilon)})$, we obtain the following a technical tool, which will prove to be useful in our endeavor.  
	
	\begin{corollary}\label{CRCW-Kar}
		Any CRCW-PRAM algorithm using $M$ total memory, $P$ processors and $T$ time
		can be simulated in $O(T\log_{N^{1-\epsilon}}P)$ 
		$\DMRC$-rounds if $(M+P)\log_{N^{1-\epsilon}}(M+P)\in O(N^{2(1-\epsilon)})$.
	\end{corollary}

	\subsection{Simulating \texorpdfstring{$\NC^1$}{NC1}}\label{Sec31}
	It is known that $\text{Nonuniform-}\NC^1$ is equal to the class of languages recognized by nonuniform width-bounded branching programs. A careful inspection of the proof due to Barrington \cite{B}---crucially relying on the non-solvability of the permutation group on 5 elements---reveals that it naturally translates to the uniform analogue: Our uniform class $\NC^1$ is identical with the class of languages recognized by uniform width-bounded branching programs. In order to prove $\NC^1 \subseteq \DMRC^0$, it therefore suffices to show how to simulate such branching programs by appropriate MapReduce computations with a constant number of rounds. 
	
	We first define width-bounded branching programs. Let $n,w\in\N_+$. The input to the program is an assignment $\alpha$ to $n$ Boolean variables $\mathcal{X}=\{x_0,\ldots,x_{n-1}\}$. An {\em instruction} or {\em line} of the program is a triple $(x_i,f,g)$, where $i$ is the \emph{index} of an input variable $x_i\in\mathcal{X}$ and $f$ and $g$ are endomorphisms of $\range{w}$. An instruction $(x_i,f,g)$ \emph{evaluates} to $f$ if $\assignment(x_i)=1$ and to $g$ if $\assignment(x_i)=0$. 
	A {\em width-$w$ branching program} 
	of length $t$ is a sequence of instructions $(x_{i_j},f_j,g_j)$ for $j\in\range{t}$. 
	We also refer to the $t$ instructions as the lines of the program. 
	Given an assignment $\assignment$ to $\mathcal{X}$%
	, a branching program $\BranchingProgram$ yields a function $\BranchingProgram(\assignment)$ that is the composition of the functions to which the instructions evaluate. 
	
	To recognize a language $L \subseteq \{0,1\}^*$, we need a family $(\BranchingProgram_n)_{n=0}^\infty$ of width-$w$ branching programs with $\BranchingProgram_n$ taking $n$ Boolean inputs. 
	We say that $L$ is recognized by $\BranchingProgram_n$ if there is, for each $n\in\N$, a set $F_n$ of endomorphisms of $\range{w}$ such that for all $\assignment \in \{0,1\}^n$, $\assignment \in L$ if and only if $\BranchingProgram_n(\assignment) \in F_n$.
	If $f_i$ and $g_i$ are automorphisms, that is, permutations of $\range{w}$ for all $i\in\range{t}$, then $\BranchingProgram_n$ is called a {\em width-$w$ permutation branching program}, or $w$-\PBP{} for short. 
	
	\begin{theorem}\label{NCeqPBP} \cite{B}
		If $L \in \NC^1$, then $L$ is recognized by a logspace-uniform $5$-\PBP{} family.
	\end{theorem}
	
	Due to Theorem~\ref{NCeqPBP} it is sufficient for our purposes to simulate the $w$-\PBP{}s with constant $w$ instead of the circuit families provided by the definition of $\NC^1$. 
	In order to do this, we need to encode the given $w$-\PBP{} and the possible assignments in the right form, namely we express them as sets of key-value pairs.
	A $w$-\PBP{} of length $t$ can be described as the set $\{\langle\, p;\, (x_{i_p}, f_p, g_p)\,\rangle \mid p\in\range{t}\}$, where we call $p$ the \emph{line number} of line $(x_{i_p}, f_p, g_p)$. Similarly, an assignment $\assignment\colon \mathcal{X}\to \{0,1\}, x_i\mapsto v_i$ to the input variables $\mathcal{X}=\{x_0,x_1, \ldots, x_{n-1}\}$ is described by the set of key-value pairs  
	$\{\langle i; (x_i, v_i)\rangle \mid i\in\range{n} \}$, letting the mappers divide the information by the indices of the input variables. 
	Let $N_\textnormal{O}$ and $N_\textnormal{I}$ be the total size of the encodings of the $w$-\PBP{} and the input assignment $\assignment$, respectively. Let $N=N_\textnormal{O}+N_\textnormal{I}$ and
	let $d=\lceil N_\textnormal{O}^{1-\epsilon} \rceil$ and $\ell=\lceil N_\textnormal{O}^{\epsilon} \rceil$.
	We denote by $\div$ the integer division. 
	For every $q\in\range{t\div{}d}$, let $w$-\PBP\!$_q$ be the $q$th of the subprogram blocks of $w$-\PBP{} of length $d$, that is $\{\langle p; (x_{i_p}, f_p, g_p)\rangle\mid qd \leq p \leq (q+1)d-1\}$. 
	For ease of readability, we assume from now on without loss of generality that $d\ell=t$, so that $w$-$\PBP{}$ can be partitioned into exactly $\ell$ such subprograms. 
	
	For every $q\in \range{\ell}$, we denote by $\INq$ the subset of variables from $\mathcal{X}$ appearing in the instructions of subprogram $w$-\PBP\!$_q$. An assignment $\partialassignmentq\colon\INq\to\{0,1\}$ to these variables is represented as a set of key-value pairs in the following way. Recall that the subprogram $w$-\PBP\!$_q$ is a list of lines, each of which requires the assignment of a value, either 0 or 1, for exactly one variable. 
	Let $x_{q,j}$ be the $j$th variable to which a value is assigned in $w$-\PBP\!$_q$, let $p_{q,j}$ denote the %
	number of the line in which this assignment occurs for the first time in $w$-\PBP\!$_q$, and let $v_{q,j}$ denote the value that is assigned to $x_{q,j}$ in this line. Now, we represent $\partialassignmentq$ by $\{ \langle q; (p_{q,j}, x_{q,j}, v_{q,j})\rangle\mid j\in\range{|\INq|} \}$. Note that despite the dependence of $\INq$ on $q$, we always have $|\INq|\leq d$.
	Having seen how to express $w$-\PBP{}, $\assignment$, and both $w$-\PBP\!$_q$ and $\partialassignmentq$ for all $q\in\range{\ell}$ as a set of key-value pairs, we are ready to state and prove the following lemma. 
	
	\begin{lemma}\label{PBPsim}
		Let $L$ be a $w$-\PBP{}-recognized language. %
		If, for every $q\in\range{\ell}$, the representations of $w$-\PBP{} and $\partialassignmentq$ are given, then we can decide in a $2$-round $\DMRC$-computation whether $\alpha\in L$ or not.
	\end{lemma}
	
	\begin{proof}
		As already described above, let $w$-\PBP{} be represented by the set $\{\langle p; (x_{i_p}, f_p, g_p)\rangle \mid p\in\range{t}\}$ and, for every $q\in\range{\ell}$, the assignment $\partialassignmentq$ by $\{\langle q, (p_{q,j},x_{q,j}, v_{q,j})\rangle \mid j\in\range{|\INq|}\}$. Note that there are $\ell$ subprograms of length at most $d$ and $\ell$ partial assignments that each assign values to at most one variable per line of the corresponding partial program, thus the total size of the input is in $O(d\ell)\subseteq O( N_\textrm{O})\subseteq O(N)$. 
		
		We define the first map function $\mu_1$ by
		\begin{align*}
			&\mu_1(\langle p; (x_{i_p}, f_p, g_p)\rangle) &&= 
			\ \{\langle\,\ p\div{}d\,;\,(p, x_{i_p}, f_p, g_p)\,\rangle\},\textnormal{ for each }p\in\range{t}\textnormal{ and}\\
			&\mu_1(\langle q; (p_{q,j},x_{q,j}, v_{q,j})\rangle)&&= 
			\ \{\langle\,p_{q,j}\div{}d\,;\, (p_{q,j},x_{q,j}, v_{q,j})\,\rangle\}\textnormal{ for each }q\in\range{\ell}, j\in\range{k+1}.
		\end{align*}
		For any $q\in\range{\ell}$, there is one subprogram $w$-\PBP\!$_q$ and an associated assignment set $\partialassignmentq$.  
		We use the map function $\mu_1$ to find the value assignment for each variable appearing in $w$-\PBP\!$_q$ and store it in a key-value pair. This pair has the key $q$ and is thereby designated to be processed by reducer$_q$, which can calculate $\rho_1$, having all pairs with key $q$ available. This function simulates, for each permutation $\pi$ of $\range{w}$, the subprogram $w$-\PBP\!$_q$ on this permutation with the received assignment and stores the resulting permutation $\pi'$. This yields a table $T_q$ of size $w!\in O(1)$, describing the action of $w$-\PBP\!$_q$ for the given assignment on all $w!$ permutations. (We mention in passing that for the first reducer$_0$ it would be sufficient to compute and store only the permutation that results from applying $w$-\PBP{}$_0$ on the given assignment to the identity as the initial permutation, thus saving the time and memory necessary for the rest of the first table.) The output of $\rho_1$ on the $q$th reducer is $\langle q; T_q\rangle$.
		
		The map function $\mu_2$ of the second round is simple, it maps $\langle q; T_q\rangle$ to $\langle 0; (q,T_q)\rangle$, thus delivering all pairs $(i,T_i)$ to a single instance of the reduce function $\rho_2$. This first reducer has therefore all tables $T_0,\ldots,T_{\ell-1}$ at its disposal and knows which one is which.
		Using $T_q$ as a look-up table for the permutation performed by $w$-\PBP\!$_q$, reducer$_0$ can now compute, starting from the identity permutation $\textnormal{id}$, the permutation $\pi=T_{\ell-1}\circ \dots \circ T_2\circ T_1\circ T_0(\textnormal{id})$, and the input is accepted if and only if $\pi \in F_n$, where $F_n$ is the set of accepted permutations that is given to us alongside the  program $w$-\PBP{}.
	\end{proof}
	
	In the following four lemmas, we show that $\partialassignmentq$ can be computed in a constant number of rounds from $w$-\PBP{} and $\assignment$ for every $q\in\range{\ell}$. 
	The challenge lies in designing an interface between the different reducers to bridge the gap between the $\ell$ program blocks $w$-\PBP{}$_q$ and the given assignments, initially cut into $\ell$ block based solely on the indices of the input variables, without exceeding the memory limits.
	We begin with a brief overview of the four steps. 	
	
	\begin{enumerate}
		\item For each $x_i$, where $i\in\range{n}$, we compute the number of subprograms 
		in which $x_i$ appears, and denote this number by $\subprogramNumber(x_i)$. 
		Note that $\subprogramNumber(x_i)\leq \ell$ and that $\subprogramNumber(x_i)$ is the number of all those reducers for which the value assignment of $x_i$ is generally required to compute the resulting permutations in the corresponding subprograms. 
		\item We compute the prefix sums of $\subprogramNumber(x_i)$. For $i\in\range{n}$, let $y_i = \sum_{j=0}^{i} \subprogramNumber(x_j)$. Note that 
		$y_i$ is the number of assignment triples $(p_{q,j},x_{q,j},v_{q,j})$ with $0<j\le i$ needed to compute the action of the first $i$ subprograms 
		and that $y_{n-1}=\sum^{\ell-1}_{q=0}|\partialassignmentq|$. 
		\item Based on the prefix sums, we will compute a \emph{separation} of the input variables into $\ell$ contiguous blocks such that, for each $q\in\range{\ell}$, it is feasible for reducer$_q$ to produce from the $q$th block the input value assignments that it needs to contribute for the next step.  
		This is nontrivial since the number of input assignments must not exceed $O(d)$ due to the memory limitation of reducer$_q$. %
		A \emph{separation} of the input variables $\{x_0,\ldots,x_{n-1}\}$ is a list of $\ell-1$ {\em split values} $\sigma_1,\ldots,\sigma_{\ell-1}$ such that we have $\ell$ ordered, contiguous blocks  $\{x_0,\ldots,x_{\sigma_1}\}, \{x_{\sigma_1+1},\ldots,x_{\sigma_2}\},\ldots, \{x_{\sigma_{\ell-1}+1},\ldots,x_{n-1}\}$. For notational convenience, we let $\sigma_0=-1$ and $\sigma_\ell=n-1$. 
		Let $\sigma_q =\max\{j\in[n]\mid y_j\leq qd\}$ for $q\in\{1,\ldots,\ell-1\}$.
		Using these split values, %
		each reducer$_q$ can provide all value assignments needed for the computation of all subprograms in the next step without violating the memory limitations.
		
		\item We compute $\partialassignmentq$ for $q\in\range{\ell}$ by using $w$-\PBP{}, the input assignment $\assignment$, and the split values.  %
	\end{enumerate}
	
	\begin{lemma}\label{computing-SharpP}
		Calculating $\subprogramNumber(x_i)$ is in $\DMRC^0$. That is, for each $i\in\range{n}$, $\subprogramNumber(x_i)$ is computable from $w$-\PBP{} in a constant number of $\DMRC$-rounds.
	\end{lemma}
	
	\begin{proof}
		For each $q\in\range{\ell}$, the subprogram $w$-\PBP\!$_q$ is stored in reducer$_q$. The output of reducer$_q$---which will be the input to compute $\subprogramNumber(x_i)$---is $\langle q;(q,1)\rangle, \ldots, \langle q;(q,k_q)\rangle$, with the variables $x_{q,1},\ldots,x_{q,k_q}$ appearing in the subprogram $w$-\PBP\!$_q$ and $k_q \in O(d)$. The total number of inputs used  to compute $\subprogramNumber(x_i)$ is therefore at most $d\ell\in O(N)$. We use a Sum-CRCW-PRAM, whose concurrent writes to a single memory register are resolved by summing up all values being written to the same register simultaneously, see \cite{GSZ}. 
		We use at most $d\ell$ processors, $\processor_{q,1}, \ldots, \processor_{q,k_q}$ for each $q\in\range{\ell}$, and registers $\register_0, \ldots, \register_{n-1}$ and let all processors $\processor_{q,j}$ add $1$ to $\register_j$ concurrently. Thus we see that the computing $\subprogramNumber(x_i)$ is possible in constant time on a Sum-CRCW-PRAM and therefore, by Corollary~\ref{CRCW-Kar}, in $\DMRC^0$. 
	\end{proof}

	\begin{lemma}\label{prefix-MRC}
		Computing the prefix-sums of $\subprogramNumber(x_i)$ is in $\DMRC^0$. 
	\end{lemma}

	\begin{proof}
		The input is given as $\langle i; (\subprogramNumber(x_i), i) \rangle$ for $i\in\range{n}$.
		We compute the prefix-sums $y_i$ of $\subprogramNumber(x_i)$ for all $i\in\range{n}$ in three rounds that can be summarized as follows:
		\begin{enumerate}
			\item Each reducer$_q$, for $q\in\range{\ell}$, determines its local prefix-sums; that is, it computes the $d$ prefix-sums $\localy_{dq}, \ldots, \localy_{d(q+1)-1} $ of the $d$ numbers $\subprogramNumber(x_{dq}), \ldots, \subprogramNumber(x_{d(q+1)-1})$. 
			\item A single reducer computes the prefix-sums $z_{0}, z_{1}, \ldots z_{\ell-1}$ of $\localy_{d-1}, \localy_{2d-1}, \ldots \localy_{\ell d-1}$, which are known from the first round. For every $q\in\range{\ell-1}$, we send $z_q$ to reducer$_{q+1}$.%
			\item Each reducer$_{q+1}$ with $q\in\range{\ell-1}$ computes $y_{d(q+1)+j} = \localy_{d(q+1)+j} + z_q$ for each $j\in\range{d}$. 
		\end{enumerate}
		
		We now describe the three rounds in more detail at the level of key-value pairs. 
		\begin{enumerate}
			\item By defining the map function $\mu_1(\langle i; (\subprogramNumber(x_i), i)\rangle{}) = \langle i\div{}d; (\subprogramNumber(x_i), i)\rangle{}$,
			each reducer$_q$, for $q\in\range{\ell}$,  receives $\subprogramNumber(x_{dq}), \ldots, \subprogramNumber(x_{d(q+1)-1})$ together with the correct indices. Thus we can compute in reducer$_q$ all 
			local prefix-sums $\localy_{dq}, \ldots, \localy_{d(q+1)-1}$ of these number. %
			The output of reducer$_q$ consists of the local prefix-sums in the format $\langle q; (\psum, q,j, \localy_{q,j}) \rangle$ for $j\in\range{d}$ and the last of each group of local prefix-sums in the format $\langle q; (\last, \localy_{d(q+1)-1}) \rangle{}$, where $\psum=0$ and $\last=1$ is a simple binary identifier %
			
			\item By defining the map function $\mu_2(\langle q; (\last, \localy_{d(q+1)-1})\rangle{})= \langle 0; (\last, \localy_{d(q+1)-1})\rangle{}$, all last parts of the local prefix-sums can be gathered in reducer$_0$. Thus, the prefix-sums $z_{0}, z_{1}, \ldots z_{\ell-1}$ of $\localy_{d-1}, \ldots, \localy_{d\ell-1}$ can be computed in it and the output of the reducer is $\langle 0; (\last, i+1, z_{i}) \rangle$ for every $i\in\range{\ell-1}$.
			All other key-value pairs---that is, those of the form $\langle q; (\psum, q,j, \localy_{q,j}) \rangle$---are passed on unaltered.
			
			\item The input of the third round consists of the output pairs $\langle q; (\psum, q, j, \localy_{q,j})\rangle$ for all $j\in\range{d}$ and $q\in\range{\ell}$ passed on from the first round and the pairs $\langle 0; (\last, q+1, z_{q})\rangle$ for all $q\in\range{\ell-1}$ from the second round.
			Defining the map function as $\mu_3(\langle q; (\psum, q, j, \localy_{q,j})\rangle{}) = \langle q; (\psum, q, j, \localy_{q,j})\rangle$ and $\mu_3(\langle 0; (\last, q+1, z_{q})\rangle{})=\langle q+1; (\last, q+1, z_{q})\rangle{}$, we can, for each $j\in\range{d}$ and each $q\in\{1,\ldots,\ell-1\}$, compute $y_{q,j} = \localy_{q,j} + z_j$ in reducer$_q$.
		\end{enumerate}
		
		The memory limitations of the mappers and reducers are clearly respected.
	\end{proof}

	\begin{lemma}\label{split-values-computable}
		Each of the split values $\sigma_1,\dots,\sigma_{\ell-1}$ can be computed in one reducer with the required prefix-sums being made available in one more $\DMRC$-round. 
	\end{lemma}
	
	\begin{figure}[htbp]\centering
		
		\begin{tikzpicture}[scale=.9]
		\node (empty) at (0,-1.2) {};
		\draw (0,0)--(10,0)--(0,10)--(0,0);
		\draw (0,.75)--(9.25,.75);
		\draw (0,2)--(8,2);
		\draw (0,4)--(6,4);
		\draw (0,6.25)--(3.75,6.25);
		\draw (0,7)--(3,7);
		\draw (0,8.5)--(2-0.5,8.5);
		\draw (0,9.25)--(1-0.25,9.25);
		
		\draw (0.75,0)--(0.75,10-0.75);
		\draw (1.5,0)--(1.5,10-1.5);
		\draw (3,0)--(3,10-3);
		\draw (3.75,0)--(3.75,10-3.75);
		\draw (9.25,0)--(9.25,10-9.25);
		
		\draw[dashed,white] (0,8.25)--(0,7.25);
		\draw[dashed,white] (0.75,8.25)--(0.75,7.25);
		\draw[dashed,white] (1.5,8.25)--(1.5,7.25);
		\draw[dashed,white] (1.75,8.25)--(2.75,7.25);
		\draw[dashed,white] (0,6)--(0,4.25);
		\draw[dashed,white] (0.75,6)--(0.75,4.25);
		\draw[dashed,white] (1.5,6)--(1.5,4.25);
		\draw[dashed,white] (3,6)--(3,4.25);
		\draw[dashed,white] (3.75,6)--(3.75,4.25);
		\draw[dashed,white] (4,6)--(5.75,4.25);
		\draw[dashed,white] (0,3.75)--(0,2.25);
		\draw[dashed,white] (0.75,3.75)--(0.75,2.25);
		\draw[dashed,white] (1.5,3.75)--(1.5,2.25);
		\draw[dashed,white] (3,3.75)--(3,2.25);
		\draw[dashed,white] (3.75,3.75)--(3.75,2.25);
		\draw[dashed,white] (6.25,3.75)--(7.75,2.25);
		\draw[dashed,white] (0,1.75)--(0,1);
		\draw[dashed,white] (0.75,1.75)--(0.75,1);
		\draw[dashed,white] (1.5,1.75)--(1.5,1);
		\draw[dashed,white] (3,1.75)--(3,1);
		\draw[dashed,white] (3.75,1.75)--(3.75,1);
		\draw[dashed,white] (8.25,1.75)--(9,1);
		
		\draw[dashed,white] (4,0)--(9,0);
		\draw[dashed,white] (4,0.75)--(9,0.75);
		\draw[dashed,white] (4,2)--(7.75,2);
		\draw[dashed,white] (4,4)--(5.75,4);
		
		\draw[dashed,white] (1.75,0)--(2.75,0);
		\draw[dashed,white] (1.75,.75)--(2.75,.75);
		\draw[dashed,white] (1.75,2)--(2.75,2);
		\draw[dashed,white] (1.75,4)--(2.75,4);
		\draw[dashed,white] (1.75,6.25)--(2.75,6.25);
		\draw[dashed,white] (1.75,7)--(2.75,7);
		
		\draw[dashed,white] (1.75,4)--(2.75,4);
		
		\node[left=.8ex] at (0,0) {$x_{n-1}=x_{\sigma_\ell}$};
		\node[left=.8ex] at (0,.75) {$x_{n-2}$};
		\node[left=.8ex] at (0,2) {$x_{\sigma_{\ell-1}}$};
		
		\node[left=.8ex] at (0,4) {$x_{\sigma_2}$};
		\node[left=.8ex] at (0,6.25) {$x_{\sigma_1+1}$};
		
		\node[left=.8ex] at (0,7) {$x_{\sigma_1}$};
		
		\node[left=.8ex] at (0,8.5) {$x_2$};
		\node[left=.8ex] at (0,9.25) {$x_1$};
		\node[left=.8ex] at (0,10) {$x_0=x_{\sigma_0+1}$};
		
		\node[below=1ex] at (0,0) {};
		\node[below=1ex] at (.75,0) {$y_0$};
		\node[below=1ex] at (1.5,0) {$y_1$};
		\node[below=1ex] at (2.25,0) {};%
		\node[below=1ex] at (3.1,0) {$y_{\sigma_1}$};
		\node[below=1ex] at (4.3,0) {$y_{\sigma_1+1}$};
		\node[below=1ex] at (6.6,0) {$\cdots$};
		\node[below=1ex] at (9.1,0) {$y_{n-2}$};
		\node[below=1ex] at (10.5,0) {$y_{n-1}$};
		
		\draw [decorate,decoration={brace}](0.2,10)--(3.3,6.9)node [midway,xshift=2.2em,yshift=.6em] {$\textnormal{reducer}_0$}; %
		\draw [decorate,decoration={brace}](3.75,6.45)--(6.2,4)node [midway,xshift=2.2em,yshift=.6em] {$\textnormal{reducer}_1$}; %
		\draw [decorate,decoration={brace}](8,2.2)--(10.2,-0)node [midway,xshift=2.9em,yshift=.6em,rotate=0] {$\textnormal{reducer}_{\ell-1}$};  %
		
		\end{tikzpicture}
		
		\caption{Separation of the input variables $x_0,\ldots,x_{n-1}$ into $\ell$ blocks for the $\ell$ reducers, in dependence of the values of $y_i$.}\label{fig-separations}
	\end{figure}
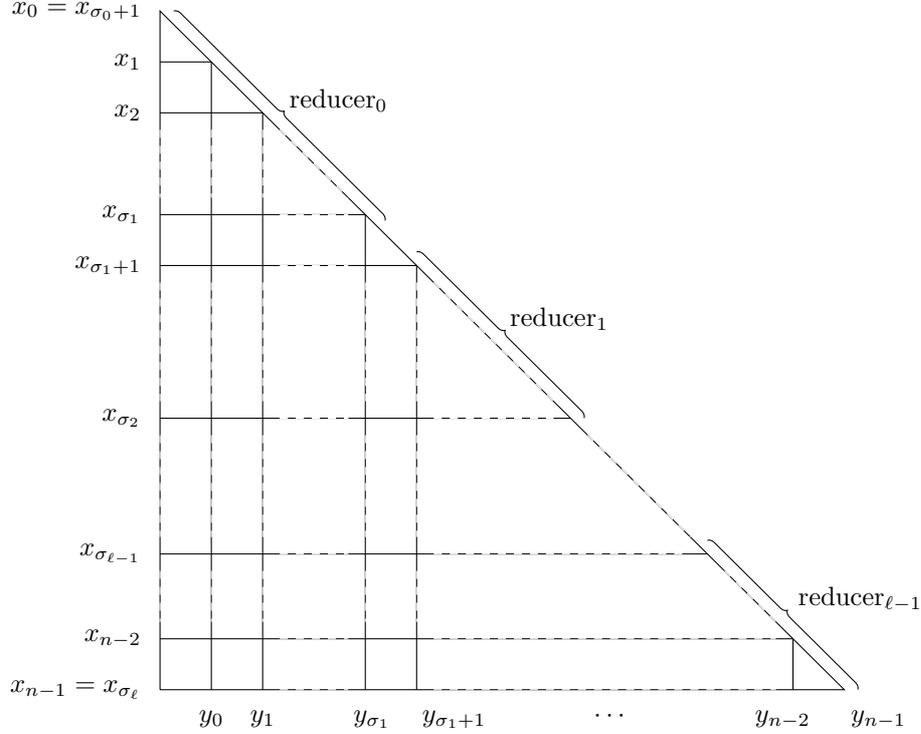
	
	\begin{proof}
		If there is a $k\in\range{\ell-1}$ such that $y_{n-1} \leq k d$, then it is clear from the definition $\sigma_q =\max\{j\in[n]\mid y_j\leq qd\}$ of the split values that $\sigma_k=\sigma_{k+1}=\ldots=\sigma_{\ell-1}$. 
		We can therefore assume that $y_{n-1} > (\ell-1)d$ and characterize, for each $q\in\{1,\dots,\ell-1\}$, the split value $\sigma_q$ as the unique integer satisfying $(q-1)d < y_{\sigma_q} \leq qd\text{ and }qd < y_{\sigma_q+1}$; see Figure~\ref{fig-separations}. 
	
		This characterization is well defined since $0<\subprogramNumber(x_i) \leq \ell < d$ for each $i\in\range{n}$ and $y_{n-1} \leq d\ell\in O(N_\textnormal{O})$. %
		For each $q\in\range{\ell}$, in order to determine the split value $\sigma_q$, it is therefore sufficient to have available in the respective reducer a sequence of consecutive prefix-sums such that the first one is at most $qd$ and the last one is greater than $qd$. 
		This condition is satisfied if reducer$_q$ has the $d+2$ consecutive prefix-sums
		$y_{qd-1}, y_{qd}, \ldots, y_{(q+1)d-1}, y_{(q+1)d}$ available. (For the first and the last reducer, the $d+1$ prefix-sums $y_{0}, \ldots, y_{d-1}, y_{d}$ and $y_{(\ell-1)d-1}, y_{(\ell-1)d}, \ldots, y_{\ell d-1}$, respectively, will suffice.) Slightly extending the sequence of available prefix-sums in each reducer by copying the overlapping prefix-sums from another reducer thus enables us to compute all split values in the $\ell$ reducers.
		Since for each $q\in\range{\ell}$, there are the $d$ prefix-sums $y_{qd}, \ldots, y_{(q+1)d-1}$ in reducer$_q$, each reducer can have the $d+2$ prefix-sums made available after one more round by having each neighboring reducer copy one more prefix-sum into it. 
		We have $\sigma_0=-1$ and $\sigma_{\ell}=n-1$; it is thus immediately verified that, for every $q\in\range{\ell}$, the total number of subprograms in which input variables between $x_{\sigma_q+1}$ and $x_{\sigma_{(q+1)}}$ appear is at most $2d$, showing that all the memory restrictions on the reducers are observed. %
	\end{proof}

	\begin{lemma}\label{PBPtoINq}
		Given $w$-\PBP{}, $\assignment$, and the split values $\sigma_0, \ldots, \sigma_{\ell}$,  %
		we can, for each $q\in\range{\ell}$, compute 
		$\partialassignmentq$ in a constant number of  $\DMRC$-rounds.
	\end{lemma}

	\begin{proof}
		We can assume that, for each $\kappa\in\range{\ell}$, the reducer$_\kappa$ has the subprogram $w$-\PBP\!$_\kappa$, the $\kappa$th block of input assignments $\{ ( x_j, v_j)\mid \kappa\cdot d \leq j \leq (\kappa+1)d-1 \}$, %
		and the split values $\sigma_0, \ldots, \sigma_{\ell}$ available. %
		The output of reducer$_\kappa$ then consists of the following: 
		\begin{enumerate}
			\item $\langle \kappa; (q, p, x_{i_p},f_p, g_p)\rangle{}$ for each line $(p, x_{i_p},f_p, g_p)$ in $w$-\PBP\!$_\kappa$, where $\sigma_{q}+1 \leq i_p \leq \sigma_{q+1}$.
			\item $\langle \kappa; (q, x_j,v_j)\rangle$ for each value assignment $(x_j,v_j)$ with $\sigma_{q}+1 \leq j \leq \sigma_{q+1}$.
		\end{enumerate}
		For any $\kappa\in\range{\ell}$, we need to bound the total number of outputs with key $\kappa$ from above.
		From the definition of the split values we see that this number is in $O(d)$ since it is bounded by the number of lines, which is at most $2d$, plus the number of assignments, which is at most $d$. 
		
		Naturally, the map function $\mu$ of the next round is defined by 
		\begin{enumerate}
			\item $\mu(\langle \kappa; (q, p, x_{j_p},f_p, g_p)\rangle{})= \langle q; (p, x_{j_p},f_p, g_p)\rangle{}$ and
			\item $\mu(\langle \kappa; (q, x_j,v_j)\rangle{})=\langle q; (x_j,v_j)\rangle{}$.
		\end{enumerate}
		For any $\kappa\in\range{\ell}$, the assignment variables $\partialassignmentq$ can be computed by the subsequent reduce function using the key-value pairs produced above. 
		For each $q\in\range{\ell}$, the reducer$_q$ has now available the lines of $w$-\PBP{} and the value assignments for the input variables between $x_{\sigma_{q}+1}$ and $x_{\sigma_{q+1}}$. 
		It can therefore go through all the program lines and determine, on the one hand, which value assignments they require and, on the other hand, to which subprogram they belong. 
		To required assignment information is then sent to the respective reducers by outputting $\langle q; (p \div{}d, p, x_{i_p},v_{i_p})\rangle{}$. 
	\end{proof}

	We finally obtain the desired inclusion by applying Theorem~\ref{NCeqPBP} and Lemmas~\ref{PBPsim} through~\ref{PBPtoINq}.
	
	\begin{theorem}\label{NCMRC0}
		We have $\NC^1 \subseteq \DMRC^0$.%
	\end{theorem}
	
	\subsection{Simulating \texorpdfstring{$\NC^i$ For All $i\ge2$}{NCi, i >= 2}}\label{Sec32}
	
	For the higher levels in the hierarchy of Nick's class, we show how to simulate the involved circuits directly. We begin with a short outline of the proof. 
	
	Let $C_n=(V_n,E_n)$ be a $\NC^{i+1}$ circuit with an input of size $n$, given as a set of nodes and a set of directed edges, together with an input assignment $\assignment$. 
	The total size of $C_n$ in bits is $N_\textnormal{O}$, the total size of the input assignment in bits is $N_\textnormal{I}$, and $N = N_\textnormal{O}+N_\textnormal{I}$. 
	Note that $\size(C_n)$ is polynomial in $n$ and $\depth(C_n)\in O(\log^i n)$. 
	We will take the following steps to simulate the circuit $C_n$ with deterministic MapReduce computations: 
	
	\begin{enumerate}
		\item We compute the level of each node in $C_n$. 
		\item The nodes and edges are sorted by their level. 
		\item Both the circuit $C_n$ and the input assignment $\assignment$ are divided equally among the reducers. 
		\item We split the circuit into subcircuits computable in a constant number of rounds. 
		\item A custom communication scheme collects and constructs the complete subcircuits. 
		\item The entire circuit is evaluated via evaluation of the subcircuits.
	\end{enumerate}
	
	Note that equal division of $C_n$ in the third step is very different from the split in the forth one, where the parts may differ radically in size. 
	Great care must be taken so as to no violate any of the memory and time restrictions, 
	necessitating two unlike partitions. 
	The subsequent steps then need to mediate between these dissimilar divisions. 
	We will show that the steps $(1)$ to $(6)$ can be computed in $O(\log n)$, $O(1)$, $O(1)$, $O(1)$, $O(\log n)$ and $O(\depth(C_n)/\log n)$ rounds, respectively, yielding the desired theorem. 
	
	\begin{theorem}\label{ACiMRCi}
		We have $\NC^{i+1} \subseteq \DMRC^i$, for all $i\in\N_+$ and all $0<\epsilon<1/2$. 
	\end{theorem}	
	
	\subsubsection{Computing The Levels}
	
	We begin by showing how to compute the level of each node in the circuit in $O(\log n)$ $\DMRC$-rounds by simulating a CRCW-PRAM algorithm. 
	(We mention in passing that this step requires more than a constant number of rounds, which prevents us from obtaining the result for $\NC^1\subseteq \DMRC^0$ by simulating the circuits directly; the separate approach from Subsection~\ref{Sec31} via Barrington's theorem is thus required for this case.)
	
	In \cite{TMN}, an algorithm is presented that computes the levels of all nodes in a directed acyclic graph can on a CREW-PRAM with $O(n+m)$ processors in $O(\log^2 m)$ time, where $n$ and $m$ are the numbers of nodes and edges in the graph, respectively. 
	The first stage of this algorithm relies partly on the computation of prefix-sums, which can be computed much more efficiently when switching to a CRCW-PRAM, as we will show below. A straightforward adaptation of the analysis in \cite{TMN}, taking into account the maximum in-degree and out-degree and separating out the computation of prefix-sums, yields the following result. 
	\begin{lemma}\label{CRCW-labeling}
		Let $G=(V,E)$ be a directed acyclic graph with $n$ nodes, $m$ edges, 
		maximum in-degree $d_\In$, and maximum out-degree $d_\Out$. The level of each node in $G$ can then be computed on a CRCW-PRAM with $P\in O(m+P_\textnormal{P-Sum}(O(m)))$ processors and time $T\in O((\log m)\cdot ( T_\textnormal{P-Sum}(O(m))+\log \max\{d_\In,d_\Out\}))$, where $P_\textnormal{P-Sum}(q)$ and $T_\textnormal{P-Sum}(q)$ denote, respectively, the number of processors and the computation time to compute the prefix-sums of $q$ numbers on a CRCW-PRAM. 
	\end{lemma}
	
	In the following lemma, we aim to lower the time and memory requirements for computing prefix-sums on a CRCW-PRAM as far as possible. 
	
	\begin{lemma}\label{prefix_sumonCRCW}
		The prefix-sums of $q$ numbers can be computed on a CRCW-PRAM with $P\in O(q\log q)$ processors and memory $M\in O(q)$ in constant time.
	\end{lemma}

	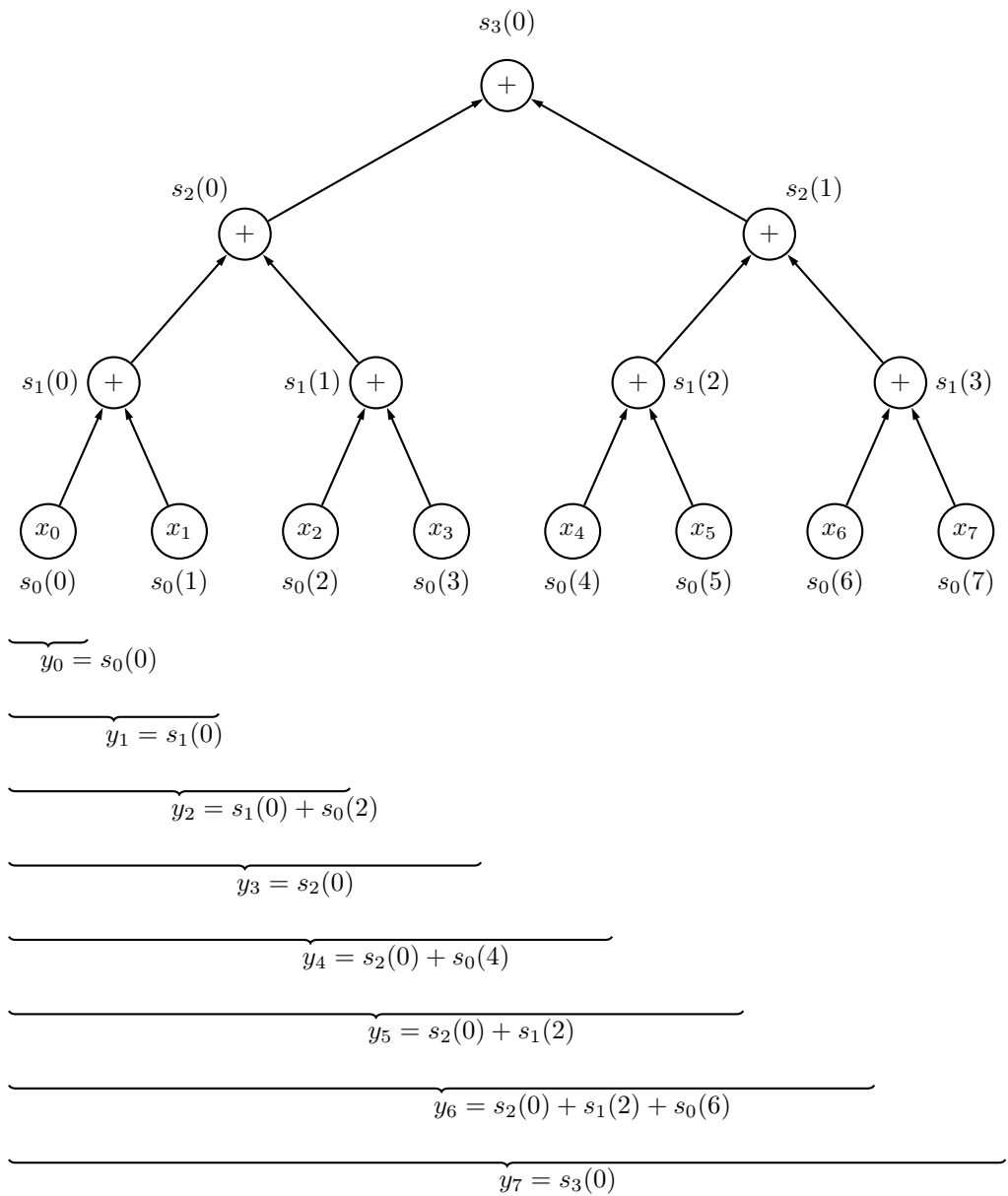
\begin{figure}[htbp]\centering
		
		\begin{tikzpicture}[xscale=1.75,yscale=1,every node/.style={draw,circle,thick},label distance=-3pt]
		\node[draw=none] (empty) at (0,-9.5) {};
		\node[label={[yshift=1.2ex]below:{$s_0(0)$}}] (00) at (0,0) {$x_0$};
		\node[label={[yshift=1.2ex]below:{$s_0(1)$}}] (01) at (1,0) {$x_1$};
		\node[label={[yshift=1.2ex]below:{$s_0(2)$}}] (02) at (2,0) {$x_2$};
		\node[label={[yshift=1.2ex]below:{$s_0(3)$}}] (03) at (3,0) {$x_3$};
		\node[label={[yshift=1.2ex]below:{$s_0(4)$}}] (04) at (4,0) {$x_4$};
		\node[label={[yshift=1.2ex]below:{$s_0(5)$}}] (05) at (5,0) {$x_5$};
		\node[label={[yshift=1.2ex]below:{$s_0(6)$}}] (06) at (6,0) {$x_6$};
		\node[label={[yshift=1.2ex]below:{$s_0(7)$}}] (07) at (7,0) {$x_7$};
		
		\node[label=left:{$s_1(0)$}] (10) at (0.5,2) {$+$};
		\node[label=left:{$s_1(1)$}] (11) at (2.5,2) {$+$};
		\node[label=right:{$s_1(2)$}] (12) at (4.5,2) {$+$};
		\node[label=right:{$s_1(3)$}] (13) at (6.5,2) {$+$};
		
		\node[label=above left:{$s_2(0)$}] (20) at (1.5,4) {$+$};
		\node[label=above right:{$s_2(1)$}] (21) at (5.5,4) {$+$};
		
		\node[label=above:{$s_3(0)$}] (30) at (3.5,6) {$+$};
		
		\draw[-Tip,thick] (00) to (10);
		\draw[-Tip,thick] (01) to (10);
		\draw[-Tip,thick] (02) to (11);
		\draw[-Tip,thick] (03) to (11);
		\draw[-Tip,thick] (04) to (12);
		\draw[-Tip,thick] (05) to (12);
		\draw[-Tip,thick] (06) to (13);
		\draw[-Tip,thick] (07) to (13);
		
		\draw[-Tip,thick] (10) to (20);
		\draw[-Tip,thick] (11) to (20);
		\draw[-Tip,thick] (12) to (21);
		\draw[-Tip,thick] (13) to (21);
		
		\draw[-Tip,thick] (20) to (30);
		\draw[-Tip,thick] (21) to (30);
		
		\clip (-.5,0) rectangle (7.5,-9.2);
		\draw[decorate,decoration={brace},thick] (0.3,-1.455)--(-0.3,-1.455) node [draw=none,midway,yshift=-1.9ex] {$\phantom{\ s_0(0)={}}y_0=s_0(0)$};
		\draw[decorate,decoration={brace},thick] (1.3,-2.455)--(-0.3,-2.455) node [draw=none,midway,yshift=-1.9ex] {$\phantom{\ s_1(0)={}}y_1=s_1(0)$};
		\draw[decorate,decoration={brace},thick] (2.3,-3.455)--(-0.3,-3.455) node [draw=none,midway,yshift=-1.9ex] {$\phantom{\ s_1(0)+s_0(2)={}}y_2=s_1(0)+s_0(2)$};
		\draw[decorate,decoration={brace},thick] (3.3,-4.455)--(-0.3,-4.455) node [draw=none,midway,yshift=-1.9ex] {$\phantom{\ s_2(0)={}}y_3=s_2(0)$};
		\draw[decorate,decoration={brace},thick] (4.3,-5.455)--(-0.3,-5.455) node [draw=none,midway,yshift=-1.9ex] {$\phantom{\ s_2(0)+s_0(4)={}}y_4=s_2(0)+s_0(4)$};
		\draw[decorate,decoration={brace},thick] (5.3,-6.455)--(-0.3,-6.455) node [draw=none,midway,yshift=-1.9ex] {$\phantom{\ s_2(0)+s_1(2)={}}y_5=s_2(0)+s_1(2)$};
		\draw[decorate,decoration={brace},thick] (6.3,-7.455)--(-0.3,-7.455) node [draw=none,midway,yshift=-1.9ex] {$\phantom{\ s_2(0)+s_1(2)+s_0(6)={}}y_6=s_2(0)+s_1(2)+s_0(6)$};
		\draw[decorate,decoration={brace},thick] (7.3,-8.455)--(-0.3,-8.455) node [draw=none,midway,yshift=-1.9ex] {$\phantom{\ s_3(0)={}}y_7=s_3(0)$};
		\end{tikzpicture}
		
		\caption{Calculation of the prefix-sums $s_i(j) = \sum_{p\in\range{(j+1)2^i}\setminus\range{j2^i}} x_p$ for every $i\in\range{1+\log q}$ and $j\in\range{q/2^i}$ for the example of $q=8$.}\label{Fig-prefixsums}

	\end{figure}
	
	\begin{proof}
		We use a Sum-CRCW-PRAM, where concurrent writes to the same memory register are resolved by adding up all simultaneously assigned numbers. \cite{GSZ}.
		Let $q$ numbers $x_0, x_1, \ldots, x_{q-1}$ be given as input. 
		Without loss of generality, we assume $q$ to be a power of $2$ and calculate $s_i(j) = \sum_{j2^i \leq p < (j+1)2^i} x_p$ for all $i\in\range{1+\log q}$ %
		and all $j\in\range{q/2^i+1}$; see Figure~\ref{Fig-prefixsums} for an illustrating example.
		
		Since each of the $q/2^i$ elements in $s_i$ is the sum of $2^i$ elements, we can---by allocating $q$ processors for each $i\in\range{1+\log q}$---compute every 
		$s_i(j)$ in a Sum-CRCW-PRAM with $O(q \log q)$ processors and $O(1)$ time. 
		
		We now describe how the prefix-sums $y(0), y(1), \ldots, y(q-1)$ are computed from the $s_i(j)$.
		Assume first that $j+1$ is a power of $2$, that is, $j+1=2^p$. Then we have $y(j)=s_{p}(0)$, so the value has already been computed. 
		If $j+1=2^p+1$ for some $p$, then we have $y(j)=s_{p}(0)+s_{0}(2^p)$, so we need to add two summands. 
		In general, $y(j)$ can be calculated as the sum of at most $\log q-1$ known summands. 
		
		Let $a_{\log q}^ja_{(\log q)-1}^j \ldots a_0^j$ be the binary representation of $j+1$.  %
		Now, we can see that
		\begin{align*}
			y(j) ={}& s_{\log q}(0) \cdot a_{\log q}^j\\
			&+ s_{(\log q) -1}((j+1-2^{(\log q)-1})\div{}2^{(\log q)-1}) \cdot a_{(\log q)-1}^j\\
			&+\ldots\\
			&+s_1((j+1-2^{1})\div{}2^{1})\cdot a_{1}^j\\
			&+s_0((j+1-2^{0})\div{}2^{0}) \cdot a_{0}^j; 
		\end{align*} 
		that is, $y(j)$ can be computed as the sum of all $s_p((j+1-2^p) \div{}2^p)$ such that $a_p^j=1$. 
		Thus, it is sufficient to supply a maximum of $(\log q)-1$ processors for the calculation of each $y(j)$ in a second time step, and the prefix-sums can be computed on a Sum-CRCW-PRAM with $O(q \log q)$ processors in constant time.
	\end{proof}
	
	We plug in the result of Lemma~\ref{prefix_sumonCRCW} into Lemma~\ref{CRCW-labeling} and then apply it to the graph $C_n$. Since its in-degrees and out-degrees are bounded by a constant $\Delta$, we have $m\le \Delta n/2\in O(n)$. Hence we can compute the levels of the nodes of $C_n$ on a CRCW-PRAM with $P\in O(N\log N)$ processors in time $T\in O(\log n)$. By Corollary~\ref{CRCW-Kar}, we obtain the following result.

	\begin{lemma}\label{computing-level}
		Computing the levels of all nodes in $C_n$ is in $\DMRC^1$.
	\end{lemma}
	
	\begin{proof}
		From Lemmas~\ref{CRCW-labeling} and~\ref{prefix_sumonCRCW} we know that the level of each node in $C_n$ can be computed in $T\in O(\log n)$ time on a Sum-CRCW-PRAM with $P\in O(N+N\log N)$ processors. 
		Now, Corollary~\ref{CRCW-Kar} yields a MapReduce simulation of this Sum-CRCW-PRAM. We need to check that the conditions of Corollary~\ref{CRCW-Kar} are indeed all satisfied: 
		From $T\in O(\log n)$, $P\in O(N+N\log N)$, and $M\in O(N)$ follows $M+P\in O(N\log N)$ and $\log_{N^{1-\epsilon}}(M+P)\in O(1)$, hence we have $(M+P)\log_{N^{1-\epsilon}}(M+P)\in O(N^{2(1-\epsilon)})$.
		Thus, the level of each node in $C_n$ can be computed in $O(\log n)$ $\DMRC$-rounds. 
	\end{proof}

	\subsubsection{Sorting By Levels}
	Once the levels of all nodes are computed, each node in the circuit can be represented as $(\level(x_i), x_i)$. %
	Recall that the depth of $C_n$ is just the maximum level. 
	Since $\depth(C_n)\in O(\log^k n)$ for some $k\in\N_+$ and the number of nodes is bounded by the number of edges $\size(C_n)\in O(N)$, we can encode each pair $(\level(x_i), x_i)$
	by appending to a bit string of length $\log(c_1\log^kn)$ another one of  $\log(c_2N)=\log(cN\log^kn)$, for appropriate constants $c_1$ and $c_2$, which results in a bit string of length $\lg(cN\log^kn)$ for $c=c_1c_2\in\N$. 
	This enables us to identify each pair $(\level(x_i), x_i)$ with a different bit string, which can interpreted as an integer bounded by $cN\log^kn$. 
	We call this integer the \emph{sorting index} of node $x_i$.  
	Crucially, we chose the bit string to start with the encoding of the level. 
	Sorting the sorting indices thus means to sort the nodes of $C_n$ by their level. 
	The following lemma shows how prefix-sums can be used to perform such a sort so efficiently on a CRCW-PRAM that we can apply Corollary~\ref{CRCW-Kar} to simulate it in a constant number of $\DMRC$-rounds. 
	
	\begin{lemma}\label{integer-sorting}
		A CRCW-PRAM with $P\in O(D\log D)$ processors and memory $M\in O(D)$ can sort any subset $I\subseteq \{1,\dots,D\}$ of integers in constant time. 
	\end{lemma}

	\begin{proof}
		Recall that we use a Sum-CRCW-PRAM that sums up concurrent writes.  
		Assume that the input and output are stored in the arrays $x[0],\ldots,x[p-1]$ and $y[0], \ldots, y[p-1]$, respectively. We will use two auxiliary arrays $z[0], \ldots,z[D]$ and $\hat{z}[0],\ldots,\hat{z}[D]$ of size $D+1$.
		The algorithm works in for steps: 
		\begin{enumerate}
			\item Initialize $z$ by using $D+1\le P$ processors to set $z[k] \leftarrow 0$ for all $k\in\range{D+1}$.
			\item Use $p\le P$ processors in parallel to set $z[x[k]] \leftarrow 
			1$ 
			for all $k\in\range{p}$. 
			
			\item Compute the prefix-sums of the array $z$ and save them into $\hat{z}$.
			
			\item Use $D$ processors to set, for all $k\in\{1,\dots,D\}$ in parallel, 
			$y[\hat{z}[k]]\leftarrow k$ 
			if and only if  $\hat{z}[k] \neq \hat{z}[k-1]$.
		\end{enumerate}
		
		Since the prefix-sums of $D$ numbers can be computed by the Sum-CRCW PRAM with $P \in O(D\log D)$ processors and memory $M \in O(D)$ in constant time by Lemma~\ref{prefix_sumonCRCW}, the above algorithm stays within these bounds as well. 
		
		We now prove that this algorithm is correct. 
		First we observe that after step 2, for every $k\in\{1,\dots,D\}$, we have $z[k]=1$ if and only if one of the $p$ integers to be sorted is $k$. %
		Because $\hat z$ contains the prefix-sums of $z$, the value stored in $\hat z[k]$ hence tells us how many of the $p$ integers in $x$ are at most $k$. (Note that accordingly we always have $z[0]=\hat z[0]=0$.) 
		Thus $k$ is one of the integers in $x$ if and only if $\hat z[k]=\hat z[k-1]+1$; otherwise, we have $\hat z[k]=\hat z[k-1]$. %
		As a consequence, the array $\hat z$ contains exactly the indices of $x$, namely $\range{p}$, as values in non-decreasing order, that is, $0=\hat z[0]\le \hat z[1]\le\dots\le \hat z[D-1]\le \hat z[D]=p$. 
		Stepping through $\hat z$ from start to end, that is, from $k=0$ to $k=D$, we therefore observe an increment of 1 from $\hat z[k-1]$ to $\hat z[k]$ exactly if $k$ one of the integers to be sorted. This means that in step 4 the integers contained in $x$ are detected from left to right in ascending order and subsequently stored into $y$ in the same order. 
	\end{proof}

	Combining Lemma~\ref{integer-sorting} and 
	Corollary~\ref{CRCW-Kar} we obtain, by a careful analysis using $\epsilon\neq 1/2$, the promised result. 
	
	\begin{corollary}\label{Sort-Kar}
		Let $c\in\N$ and $0<\epsilon<1/2$. Any set of distinct integers from $\{1,\dots,\lceil cN\log^kn \rceil\}$ can be sorted in a constant number of $\DMRC$-rounds. 
	\end{corollary}
	
	\begin{proof}
		We apply Lemma~\ref{integer-sorting} with $D\in O(N\log^k n)$. 
		We have $D\in O(N\log^k N)\subseteq O(N^{1+\zeta})$ and thus also $D\log D\in O(N^{1+\zeta})$ for any constant $\zeta>0$. 
		Choose any $\zeta<1-2\epsilon$, which is possible for $\epsilon<1/2$. The sorting is then possible on a CRCW-PRAM with $O(N^{1+\zeta})$ processors and $O(N^{1+\zeta})$ memory in constant time. 
		By Corollary~\ref{CRCW-Kar}, this CRCW-PRAM can be simulated in a constant number of $\DMRC$-rounds because $\log_{N^{1-\epsilon}}(N^{1+\zeta})=(1+\zeta)/(1-\epsilon)\in O(1)$ and $O(N^{1+\zeta})\subseteq O(N^{2(1-\epsilon)})$.
	\end{proof}
	
	Once all the nodes are sorted by their sorting index (and therefore implicitly by their level), we can enumerate them in ascending order using a \emph{sorting index} $j$; that is, we store each node as the key-value pair $\langle j; (\level(v),v)\rangle{}$. 
	Clearly, we obtain an analogous representation of the edges in the form 
	$\langle i; 
	((j, (\level(v),v),(j', (\level(v'),v'))%
	\rangle{}$%
	, which will prove useful later on. 

	\subsubsection{Division of Circuit And Assignment Among Reducers}
	
	As we have already seen when discussing the branching programs, an assignment $\assignment$ to input variables $\mathcal{X}=\{x_0,x_1, \ldots, x_{n-1}\}$ can be represented as a set 
	$\{ \langle i; (x_i, v_i)\rangle\mid i\in\range{n}\}$ of key-value pairs, where $\assignment(x_i)=v_i \in \{0,1\}$. 
	
	The circuit $C_n$ is now divided into $\ell=N_\textnormal{O}^{\epsilon}$ subsets of edges according to the sorting indices and input values that are assigned to each subset as in the case of branching programs.
	For every $q\in\range{\ell}$, %
	let $C_n^q=\{((j, \level(v),v), (j',\level(v'),v'))\mid qd \leq j \leq (q+1)d-1\}$, where $d=N_\textnormal{O}^{1-\epsilon}$, be the $q$th subset. 
	Note that $|C_n^q|\in O(d)$.
	For every $q\in\range{\ell}$, the set of variables appearing in $C_n^q$ is denoted as $\INq$ and the assignment $\partialassignmentq$ to $\INq$ is represented as
	$\{\langle j; x_{q,j}, v_{q,j}\rangle\mid j\in\range{|\partialassignmentq|} \}$, where 
	$x_{q,j}$ is the $j$th variable
	that appears as an input in $C_n^q$, and $v_{q,j}$ is its assignment value. 
	Just as seen in Lemma~\ref{PBPtoINq} for the case of a branching program, we can now compute $\partialassignmentq$ from $C_n$ and $\assignment$ for all $q\in\range{\ell}$, yielding the following lemma.
	
	\begin{lemma}
		Computing $\partialassignmentq$ from $C_n$ and $\assignment$ is in $\DMRC^0$ for every $q\in\range{\ell}$.%
	\end{lemma}
	
	We can therefore assume that each input node is represented by $\langle j; (\level(x_{j_i}),x_{j_i}, v_{j_i})\rangle$, a key-value pair that is computed from $C_n^q$ and $\partialassignmentq$ for $q\in\range{\ell}$ in a single $\DMRC$-round.

	\subsubsection{Division Into Subcircuits By Levels}
	We divide $C_n=(V_n,E_n)$ into as few subcircuits as possible such that the simulation of each subcircuit is in $\DMRC^0$ and we can evaluate $C_n$ by evaluating the subcircuits sequentially. 
	
	Given $v\in V_n$ and $\delta\in\N$, %
	we define the {\em $v$-down-circuit} $C_{\delta}^{\down}(v)=(V_{\delta}^{\down}(v), E_{\delta}^{\down}(v))$ of depth $\delta$ %
	to be the subcircuit of $C_n$ induced by 
	$V_{\delta}^{\down}(v)=\{ u\mid \level(v) \leq \level(u) \leq \level(v)+\delta, u \rightarrow^* v \}$, 
	where $u \rightarrow^* v $ means that there is a directed path of any length (including $0$) from $u$ to $v$ in $C_n$. 
	The {\em $v$-up-circuit} $C_{\delta}^{\up}(v)=(V_{\delta}^{\up}(v), E_{\delta}^{\up}(v))$ of depth $\delta$ is analogously the subcircuit of $C_n$ induced by 
	$V_{\delta}^{\up}(v)=\{ u\mid \level(v)-\delta \leq \level(u) \leq \level(v), v \rightarrow^* u \}$.
	
	When dividing $C_n$ into subcircuits we have two conflicting goals. On the one hand, we want as few of them as possible, which implies that they have to be of great depth. On the other hand, we need to simulate them in MapReduce without exceeding the memory bounds. A depth in $O(\log n)$ turns out to be the right choice. 
	Let $s=\formeralpha (\log n)/\log \Delta$, where $\Delta\ge 2$ is a constant bounding 
	the maximum degree 
	of $C_n$ and $\formeralpha$ is an arbitrary constant satisfying $0<\formeralpha<1-2\epsilon$. 
	(Note that such a $\formeralpha$ exists exactly if $\epsilon<1/2$.) %
	Since a tree of depth $\subcircuitDepth$ and maximum degree bounded by a constant $\Delta$ contains at most $\sum_{i=1}^s\Delta^i$ 
	edges, %
	their size is in $O(\Delta^s)=O(n^\formeralpha)\subseteq O(N^\formeralpha)$.  Hence each reducer may contain up to  $N^{1-\epsilon}/N^\formeralpha$ %
	such subcircuits 
	without exceeding the memory constraint of $O(N^{1-\epsilon})$; see Figure~\ref{Fig-down-up-circuits-in-reducer}. We denote this number of allowed subcircuits per reducer by $\beta=N^{1-\epsilon-\formeralpha}$. 
	
	\begin{figure}[htbp]\centering
		
		\begin{tikzpicture}[xscale=1,yscale=.8]
		\draw[very thick] (0,0) rectangle (10,12);
		\draw[decorate,decoration={brace},thick] (-.2,-0.2)--(-0.2,12+0.2) node [draw=none,midway,xshift=-2em,rotate=30] {$\beta$ edges};
		\node[xshift=3em,yshift=2ex] at (0,12) {$\textnormal{reducer}_q$};
		\draw[thick] (1,5.5)--+(8,0)--+(0,4.5)--+(8,4.5)--+(0,0); 
		\node[style={align=left}] at (5.0,11) {For every key-value pair $\langle q; (j_v,\textnormal{level}(v),v)\rangle$ \\such that there is an $i\in\mathbb{N}$ with $\textnormal{level}(v)=L_i%
			$:};
		
		\node[align=center] at (5,9.4) {$v$-up-circuit of\\depth $L_{i-1}-L_i$};
		\node[draw,fill=white,thick,circle,inner sep=2pt] at (5,7.6) {$v$};
		\node[align=center] at (5,6.1) {$v$-down-circuit of\\depth $L_i-L_{i+1}$};

		\node[style={align=left}] at (5.0,4) {For every key-value pair $\langle q; (j_v,\textnormal{level}(v),v)\rangle$\\such that $\textnormal{level}(v)\neq L_i$ for all $i\in\mathbb{N}$:
		};

		\node[draw,fill=white,thick,circle,inner sep=2pt] (center) at (8,2) {$v$};
		\node[draw,fill=white,thick,circle,inner sep=2pt] (northeast) at (9,2.8) {$\phantom{v}$};
		\node[draw,fill=white,thick,circle,inner sep=2pt] (northwest) at (7,2.8) {$\phantom{v}$};
		\node[draw,fill=white,thick,circle,inner sep=2pt] (southwest) at (7,1.2) {$\phantom{v}$};
		\node[draw,fill=white,thick,circle,inner sep=2pt] (southeast) at (9,1.2) {$\phantom{v}$};
		
		\draw[-Tip,thick] (center) to (northeast);
		\draw[-Tip,thick] (center) to (northwest);
		\draw[-Tip,thick] (southwest) to (center);
		\draw[-Tip,thick] (southeast) to (center);
		
		\draw[dotted,rounded corners] (.5,2.15) rectangle (9.5,3.2);
		\draw[dotted,rounded corners] (.5,.8) rectangle (9.5,1.85);
		\node at (3.5,2.9) {$v$-up-circuit of depth 1};
		\node at (3.5,1.6) {$v$-down-circuit of depth 1};
		
		\node (empty) at (11.5,0) {};
		\node (empty) at (0,-1) {};		
		
		\end{tikzpicture}
		
		\caption{The up-circuits and down-circuits constructed in reducer$_q$, comprising up to $\beta$ edges.}\label{Fig-down-up-circuits-in-reducer}
	\end{figure}
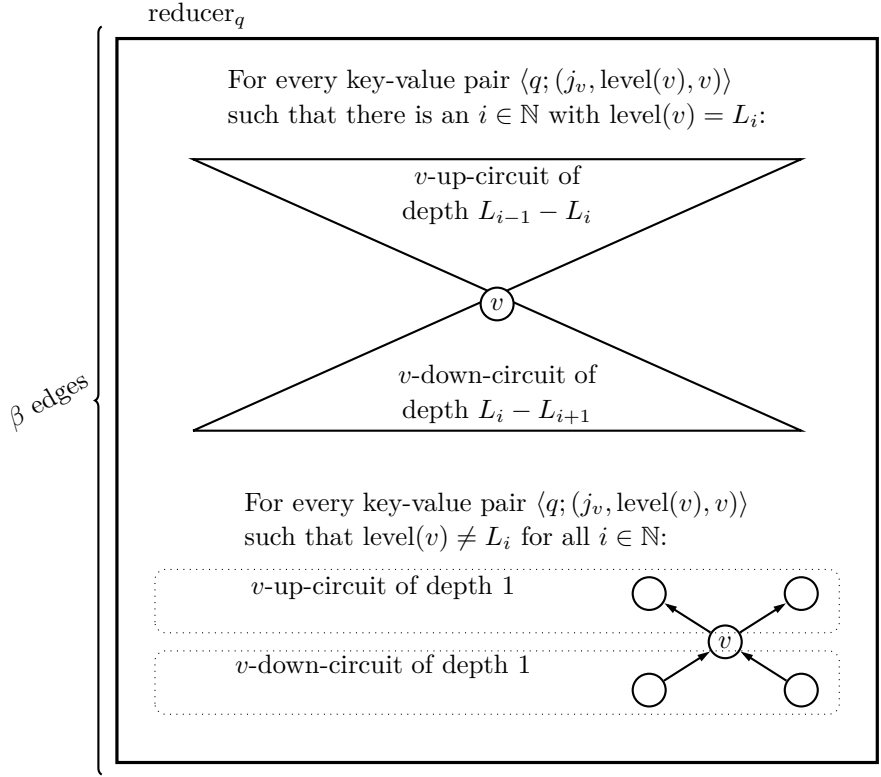
	
	For each $i\in\range{\lceil \depth(C_n)/s \rceil+1}$, we define $L_i=i\cdot s$. 
	For every node $v$ on level $L_i$---that is, with $\level(v)=L_i$---we call the $v$-down-circuit ($v$-up-circuit, resp.) of depth $\subcircuitDepth$ a \emph{$L_i$-down-circuit} (\emph{$L_i$-up-circuit}, resp.).
	We will construct in each reducer the $v$-down-circuits and $v$-up-circuits of depth $1$ of all its nodes. %
	From those we then construct all $L_i$-down-circuits and $L_i$-up-circuits for every $i$. %
	Note that we can evaluate all $L_i$-down-circuits if the values of the nodes of level $L_{i+1}$ are given. 
	The values of the nodes $v$ of level $L_{i+1}$ that are necessary to compute the $L_i$-up-circuits are then known from the $L_{i+1}$-down-circuits.
	
	When the circuit $C_n$ is divided into $L_i$-down-circuits, there may exist edges of $C_n$ that are not contained in any $L_i$-down-circuit. If an edge $((j_u, \level(u), u),(j_v, \level(v),v))$ satisfies $L_{i_u} \leq \level(u) \leq L_{i_u+1}$ and $L_{i_v} \leq \level(v) \leq L_{i_v+1}$ for $i_u \neq i_v$, then this edge is not included in any $L_{i_u}$-down-circuit nor any $L_{i_v}$-down-circuit. We call such edges {\em level-jumping edges}; see Figure~\ref{Fig-jumping-edge} for an example.
	We would like to replace every level-jumping edge $(u,v)$ by a path from $u$ to $v$ that consists only of edges that %
	will be part of the respective $L_i$-down-circuits and  $L_i$-up-circuits in the resulting, {\em augmented} circuit.
	The following lemma states that this is possible without increasing the size by too much.

	\begin{figure}[htbp]\centering
		
		\begin{tikzpicture}[xscale=1.4,yscale=1.2,every node/.style={draw,fill=black,semithick,circle, minimum size=3pt, inner sep=0pt, outer sep=0pt}]
		\coordinate (LeftUpperTriangleTop) at (0,0);
		\coordinate (LeftUpperTriangleBase) at ([yshift=-6em]LeftUpperTriangleTop);
		\coordinate (LeftUpperTriangleLeft) at ([xshift=-2.5em]LeftUpperTriangleBase);
		\coordinate (LeftUpperTriangleRight) at ([xshift=2.5em]LeftUpperTriangleBase);
		
		\coordinate (LeftMiddleTriangleTop) at (LeftUpperTriangleBase);
		\coordinate (LeftMiddleTriangleBase) at ([yshift=-4em]LeftMiddleTriangleTop);
		\coordinate (LeftMiddleTriangleLeft) at ([xshift=-3.5em]LeftMiddleTriangleBase);
		\coordinate (LeftMiddleTriangleRight) at ([xshift=3.5em]LeftMiddleTriangleBase);
		
		\coordinate (LeftLowerTriangleTop) at (LeftMiddleTriangleBase);
		\coordinate (LeftLowerTriangleBase) at ([yshift=-5em]LeftLowerTriangleTop);
		\coordinate (LeftLowerTriangleLeft) at ([xshift=-3em]LeftLowerTriangleBase);
		\coordinate (LeftLowerTriangleRight) at ([xshift=3em]LeftLowerTriangleBase);

		\coordinate (RightUpperTriangleTop) at ([xshift=12em]LeftUpperTriangleTop);
		\coordinate (RightUpperTriangleBase) at ([yshift=-6em]RightUpperTriangleTop);
		\coordinate (RightUpperTriangleLeft) at ([xshift=-2.5em]RightUpperTriangleBase);
		\coordinate (RightUpperTriangleRight) at ([xshift=2.5em]RightUpperTriangleBase);

		\coordinate (RightMiddleTriangleTop) at (RightUpperTriangleBase);
		\coordinate (RightMiddleTriangleBase) at ([yshift=-4em]RightMiddleTriangleTop);
		\coordinate (RightMiddleTriangleLeft) at ([xshift=-3.5em]RightMiddleTriangleBase);
		\coordinate (RightMiddleTriangleRight) at ([xshift=3.5em]RightMiddleTriangleBase);
		
		\coordinate (RightLowerTriangleTop) at (RightMiddleTriangleBase);
		\coordinate (RightLowerTriangleBase) at ([yshift=-5em]RightLowerTriangleTop);
		\coordinate (RightLowerTriangleLeft) at ([xshift=-3em]RightLowerTriangleBase);
		\coordinate (RightLowerTriangleRight) at ([xshift=3em]RightLowerTriangleBase);

		\coordinate (LeftOfLeftUpperTriangleTop) at ([xshift=-5em]LeftUpperTriangleTop);
		\coordinate (LeftOfLeftMiddleTriangleTop) at ([xshift=-5em]LeftMiddleTriangleTop);
		\coordinate (LeftOfLeftLowerTriangleTop) at ([xshift=-5em]LeftLowerTriangleTop);
		\coordinate (RightOfLeftUpperTriangleTop) at ([xshift=4em]LeftUpperTriangleTop);
		\coordinate (RightOfLeftMiddleTriangleTop) at ([xshift=4em]LeftMiddleTriangleTop);
		\coordinate (RightOfLeftLowerTriangleTop) at ([xshift=4em]LeftLowerTriangleTop);
		\coordinate (LeftOfRightUpperTriangleTop) at ([xshift=-4em]RightUpperTriangleTop);
		\coordinate (LeftOfRightMiddleTriangleTop) at ([xshift=-4em]RightMiddleTriangleTop);
		\coordinate (LeftOfRightLowerTriangleTop) at ([xshift=-4em]RightLowerTriangleTop);
		\coordinate (RightOfRightUpperTriangleTop) at ([xshift=6em]RightUpperTriangleTop);
		\coordinate (RightOfRightMiddleTriangleTop) at ([xshift=6em]RightMiddleTriangleTop);
		\coordinate (RightOfRightLowerTriangleTop) at ([xshift=6em]RightLowerTriangleTop);

		\coordinate (v1) at ([xshift=-.5em,yshift=1.2em]LeftLowerTriangleTop);
		\coordinate (u1) at ([xshift=-.5em,yshift=-3em]LeftLowerTriangleTop);
		\coordinate (v1right) at ([xshift=-.5em,yshift=1.2em]RightLowerTriangleTop);
		\coordinate (u1right) at ([xshift=-.5em,yshift=-3em]RightLowerTriangleTop);
		\coordinate (u1v1dum) at ([xshift=-1.2em,yshift=0em]RightLowerTriangleTop);

		\coordinate (v2) at ([xshift=.5em,yshift=1em]LeftMiddleTriangleTop);
		\coordinate (u2) at ([xshift=.5em,yshift=-4em]LeftLowerTriangleTop);
		\coordinate (v2right) at ([xshift=1em,yshift=1em]RightMiddleTriangleTop);
		\coordinate (u2right) at ([xshift=1em,yshift=-4em]RightLowerTriangleTop);
		\coordinate (u2v2dum2) at ([xshift=1.4em,yshift=0em]RightMiddleTriangleTop);
		\coordinate (u2v2dum1) at ([xshift=1.8em,yshift=0em]RightLowerTriangleTop);

		\draw[semithick] (LeftUpperTriangleLeft)--(LeftUpperTriangleRight)--(LeftUpperTriangleTop)--cycle;
		\draw[semithick] (LeftLowerTriangleLeft)--(LeftLowerTriangleRight)--(LeftLowerTriangleTop)--cycle;
		\draw[semithick] (LeftMiddleTriangleLeft)--(LeftMiddleTriangleRight)--(LeftMiddleTriangleTop)--cycle;
		\draw[semithick] (RightUpperTriangleLeft)--(RightUpperTriangleRight)--(RightUpperTriangleTop)--cycle;
		\draw[semithick] (RightMiddleTriangleLeft)--(RightMiddleTriangleRight)--(RightMiddleTriangleTop)--cycle;
		\draw[semithick] (RightLowerTriangleLeft)--(RightLowerTriangleRight)--(RightLowerTriangleTop)--cycle;

		\draw[thin] (LeftOfLeftUpperTriangleTop)--(RightOfLeftUpperTriangleTop);
		\draw[thin] (LeftOfLeftMiddleTriangleTop)--(RightOfLeftMiddleTriangleTop);
		\draw[thin] (LeftOfLeftLowerTriangleTop)--(RightOfLeftLowerTriangleTop);
		
		\draw[thin] (LeftOfRightUpperTriangleTop)--(RightOfRightUpperTriangleTop);
		\draw[thin] (LeftOfRightMiddleTriangleTop)--(RightOfRightMiddleTriangleTop);
		\draw[thin] (LeftOfRightLowerTriangleTop)--(RightOfRightLowerTriangleTop);

		\node[draw=none,fill=white,xshift=-2.5ex,yshift=0ex] at (LeftOfLeftUpperTriangleTop) {$L_i$};
		\node[draw=none,fill=white,xshift=-2.5ex,yshift=0ex] at (LeftOfLeftMiddleTriangleTop) {$L_{i+1}$};
		\node[draw=none,fill=white,xshift=-2.5ex,yshift=0ex] at (LeftOfLeftLowerTriangleTop) {$L_{i+2}$};
		
		\node[draw=none,fill=white,xshift=2.5ex,yshift=0ex] at (RightOfRightUpperTriangleTop) {$L_i$};
		\node[draw=none,fill=white,xshift=2.5ex,yshift=0ex] at (RightOfRightMiddleTriangleTop) {$L_{i+1}$};
		\node[draw=none,fill=white,xshift=2.5ex,yshift=0ex] at (RightOfRightLowerTriangleTop) {$L_{i+2}$};

		\node at (LeftUpperTriangleTop) {};
		\node at (LeftMiddleTriangleTop) {};
		\node at (LeftLowerTriangleTop) {};
		\node at (RightUpperTriangleTop) {};
		\node at (RightMiddleTriangleTop) {};
		\node at (RightLowerTriangleTop) {};
		
		\node[label=left:$v$] (v1node) at (v1) {};
		\node[label=left:$u$] (u1node) at (u1) {};
		\node[label=right:$v$] (v1rightnode) at (v1right) {};
		\node[label=left:$u$] (u1rightnode) at (u1right) {};
		
		\node[label=above:$v'$] (v2node) at (v2) {};
		\node[label=right:$u'$] (u2node) at (u2) {};
		\node[label=above:$v'$] (v2rightnode) at (v2right) {};
		\node[label=right:$u'$] (u2rightnode) at (u2right) {};
		
		\node[label={[yshift=-1.5ex,xshift=-.2em,scale=.8]left:$\dum$}] (u1v1dumnode) at (u1v1dum) {};
		
		\node[label={[yshift=-1.5ex,xshift=.2em,scale=.8]right:$\dum'_1$}] (u2v2dum1node) at (u2v2dum1) {};
		\node[label={[yshift=-1.5ex,xshift=.2em,scale=.8]right:$\dum'_2$}] (u2v2dum2node) at (u2v2dum2) {};
		
		\draw[-Tip] (u1node)--(v1node);
		\draw[-Tip] (u2node)--(v2node);
		
		\draw[-Tip] (u1rightnode)--(u1v1dumnode);
		\draw[-Tip] (u1v1dumnode)--(v1rightnode);
		\draw[-Tip] (u2rightnode)--(u2v2dum1node);
		\draw[-Tip] (u2v2dum1node)--(u2v2dum2node);
		\draw[-Tip] (u2v2dum2node)--(v2rightnode);
		
		\draw[dotted,rounded corners] (-2.7,-5.5) rectangle (1.4,.5);
		\draw[dotted,rounded corners] (2.8,-5.5) rectangle (7.2,.5);
		\draw[bend left=30,->,dashed] (1.6,-0)--(2.6,-0);
		\draw[bend left=30,->,dashed] (1.6,-1)--(2.6,-1);
		\draw[bend left=30,->,dashed] (1.6,-2)--(2.6,-2);
		\draw[bend left=30,->,dashed] (1.6,-3)--(2.6,-3);
		\draw[bend left=30,->,dashed] (1.6,-4)--(2.6,-4);
		\draw[bend left=30,->,dashed] (1.6,-5)--(2.6,-5);
		
		\node[draw=none,fill=none] (empty) at (0,-6.15) {};
		
		\end{tikzpicture}

		\caption{Two jumping edges on the left and their resolving division on the right.}\label{Fig-jumping-edge}
	\end{figure}
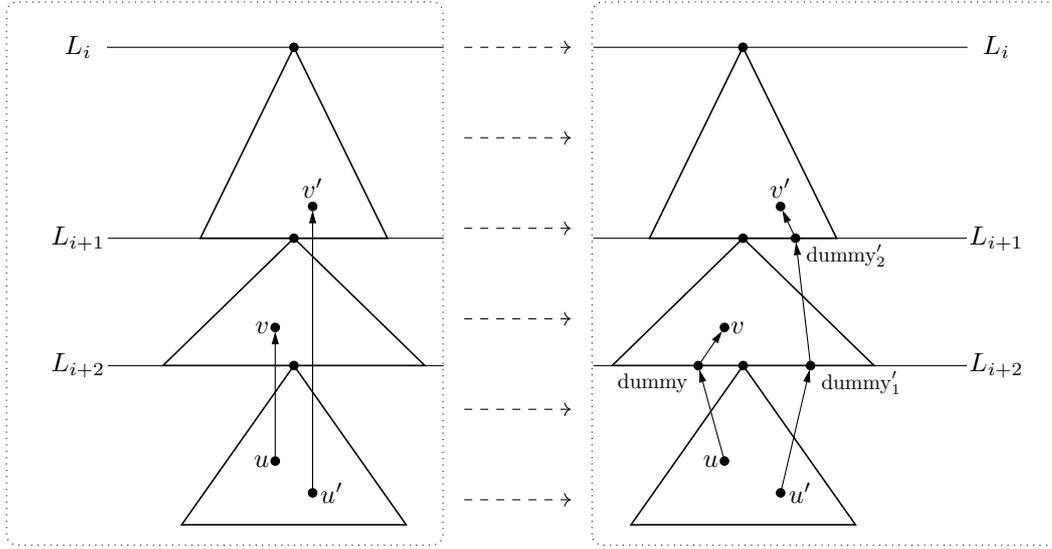

	\begin{lemma}\label{jumping-edges}
		We can subdivide the jumping edges in $C_n$ in a way that renders the subcircuit-wise evaluation possible without increasing the size beyond $O(N)$.  
	\end{lemma}

	\begin{proof}
		Let $((j_u, \level(u), u),(j_v, \level(v),v))$ be a jumping edge, where $L_{i_u} \leq \level(u) \leq L_{i_u+1}$, $L_{i_v} \leq \level(v) \leq L_{i_v+1}$ and $i_u < i_v$. If $i_u=i_v-1$, then this edge is divided into two edges $((j_u, \level(u), u), \dum)$ and $(\dum, (j_v, \level(v),v))$, introducing a new node \dum{} of the $\id$ kind with $\level(\dum)=i_v$.
		If $i_u \leq i_v-2$, then this edge is divided into three edges $((j_u, \level(u), u), \dum_1)$, $(\dum_1,\dum_2)$ and $(\dum_2, (j_v, \level(v),v))$, introducing two new nodes with $\level(\dum_1)=i_u+1$, $\level(\dum_2)=i_v$.
		Having divided the jumping edges in this way, the newly created edges are all part of some $\dum$-down-circuit or  $\dum$-up-circuit, except for edges of the form $(\dum_1,\dum_2)$. Note that we cannot further subdivide the edges of the form $(\dum_1,\dum_2)$ because we would exceed the size limit on the circuit otherwise. The most convenient way to deal with this is to adjust our definition of down-circuits and up-circuits such that every edge of the form $(\dum_1,\dum_2)$ is considered to be both a $\dum_1$-down-circuit and a $\dum_2$-up-circuit on its own. This way, every edge in the augmented circuit is included in some down-circuit or up-circuit. 
		Note that this augmentation can be performed in a single round and that the size of the augmented circuit is in $O(N)$.
		In what follows, we consider $C_n$ to be this augmented circuit. 
	\end{proof}

	\subsubsection{Construction of Subcircuits in Reducers}
	Having described the subcircuits on which the evaluation of the entire circuits will be based we now need to show how to split and construct them in the $\ell$ different reducers. 
	In each reducer, We start with the nodes $v$ contained in it that satisfy $\level(v)=L_i$ for any $i$ and the associated $v$-down-circuits and $v$-up-circuits of depth 1. We then iteratively increase the depth one by one, until the full $L_i$-down-circuits and $L_i$-up-circuits of depth up to $s$ are constructed. Note that the nodes of any level $L_i$ and their corresponding circuits may scattered across multiple reducers since edges were split equally among them according to their the sorting index and not depending on the level. 
	We therefore need to carefully implement a communication scheme that allows each reducer to encode requests for missing edges required in the construction, which are then delivered to them in multiple rounds, without exceeding any of the memory or time bounds. 
	Taking care of all these details, we obtain the following lemma. 
	\begin{lemma}\label{computing-down-up-circuits}
		Given $C_n$, all $L_i$-down-circuits and $L_i$-up-circuits can be constructed in $O(\log n)$ $\DMRC$-rounds whenever $0<\epsilon<1/2$. 
	\end{lemma}
	
	\begin{proof}
		In the first round, the map function $\mu_1$ is defined such that each reducer$_q$ is assigned (via the choice of the key) $\beta$ nodes of the form $\langle j; (\level(v),v)\rangle{}$
		and directed edges adjacent to these nodes. %
		Note that one edge can thus be assigned to two different reducers, once as an outgoing, once as in in-going edge. %
		Specifically, we define 
		\begin{align*}\mu_1(\langle\,j\,;\,(\level(v),v)\,\rangle) =
			\{\langle\,j \div{}\beta\,;\,(j,\level(v),v)\,\rangle\}
		\end{align*}
		for the key-value pairs representing nodes and
		\begin{align*}
			\mu_1(\langle i;(\,(j,\level(v),v),(j',\level(v'),v')\,)\rangle)=\{\ &\langle j\div{}\beta;\  ((j,\level(v),v),(j',\level(v'),v'))\rangle,\\
			&\langle j'\!\!\div{}\beta;\, ((j,\level(v),v),(j',\level(v'),v'))\rangle\ \}
		\end{align*}
		for the key-value pairs representing edges.
		
		In the subsequent execution of $\rho_1$, each reducer can therefore directly construct the $v$-up-circuits and $v$-down-circuits of depth $1$ for its  $\beta$ assigned nodes. 
		We will now describe how some of these initial circuits, namely those on levels $L_i$ for any $i\in\range{r}$, can be used to extended to full $L_i$-up-circuits and $L_i$-down-circuits by iteratively increasing the circuit depth one by one in the following way:

		Let $v$ be a node with $\level(v)=L_i$ in reducer$_q$ for any $i\in\range{r}$ and $q\in\range{\ell}$. 
		We want to extend $C_{1}^{\down}(v)$ and $C_{1}^{\up}(v)$ to $C_{2}^{\down}(v)$ and $C_{2}^{\up}(v)$, respectively. 
		Let $u_{\In}$ ($u_{\Out}$, resp.) be any node of in-degree (out-degree, resp.) $0$ in it, %
		that is, any node that potentially needs to be extended by one or multiple edges. 
		These extending edges are not necessarily available in reducer$_q$, however. We need to find out which reducer stores them---if there are any---and then request these edges from it in some way. 
		To determine the right reducer, we make use of the sorting index stored alongside each node, even when part of an edge. 
		Any edge $(u_\In,v)$ that we need to check for possible extensions is in fact represented as $\langle q\;,\;(\,(j_{u_\In},\level(u_\In),u_\In)\,,\,(j_v,\level(v),v)\,)\;\rangle$ in reducer$_q$. 
		The number of the reducer containing the downward extending edges is now retrieved as $\To(u_\In)=j_{u_\In}\div\beta$. 
		Analogously, the upward extending edges for an edge $(v,u_\Out)$ are to be found in reducer$_\To(u_\Out)$, where $To(u_\Out)=j_{u_\Out}\div\beta$. 
		We now know whom to ask for edges extending  the subcircuit beyond node $u$, namely reducer number $\To(u)$. Let $\From(v)=q$ denote the number of the reducer sending the request, which we encode in form of the key-value pair 
		$\langle q; (u, \To(u), \From(v))\rangle$. 
		
		Each reducer$_q$ does the above for every node with possible extending edges and also passes along to the mapper all $v$-up-circuits and $v$-down-circuits constructed so far unaltered. 
		This concludes the first round. 

		In the second round, the map function $\mu_2$ naturally re-assigns 
		$\langle q; (u, \To(u),\From(v))\rangle$ to reducer$_{\To(u)}$, and returns the $v$-up-circuits and $v$-down-circuits to the reducers that sent them. %
		Having received the edge request of the form of $\langle \To(u); (u, \To(u),\From(v))\rangle$ while executing $\rho_2$ reducer$_{\To(u)}$ now sends all edges potentially useful to reducer$_{\From(v)}$---that is, the entire 
		$u$-up-circuit and the entire $u$-down-circuit of depth $1$---to the next mapper in the form of a pair $(\From(v),e)$ for every edge containing node $u$. 
		As before, all other circuits constructed so far get passed along without modification as well.

		In the third round, the map function $\mu_3$ routes the requested edges to the requesting reducer by generating the key-value pairs $\langle \From(v); (\From(v), e\rangle{}$. In the reducing step, which implements the same reduce function $\rho_1$ as in the first round, reducer  reducer$_{\From(v)}$ now finally has all fully extend the $v$-up-circuits and $v$-down-circuits to depth 2. 

		Since performing the two rounds $\mu_2, \rho_2, \mu_3, \rho_1$ deepens the $L_i$-up-circuits and $L_i$-down-circuits by one level in the way just seen,
		the complete $L_i$-up-circuits and $L_i$-down-circuits can be constructed by repeating these two rounds $\subcircuitDepth$ times. 

		It is again clear that the memory and I/O requirements of the reducers are all met in every round since the input size and output size are in $O(d)$ for each reducer. 
		Moreover,the total memory for storing the $v$-up-circuits and $v$-down-circuits is 
		$\beta \cdot N\in O(N^{1+\formeralpha})$ because $C_n$ has at most $N_\textnormal{O}\in O(N)$ nodes. 
		Since the constant $\formeralpha$ was chosen such that $0<\formeralpha\leq 1-2\epsilon$, we have $N^{1+\formeralpha}\in O(N^{2(1-\epsilon)})$ and thus all up-circuits and down-circuits can be stored in the respective reducers. 
	\end{proof}

	\subsubsection{Evaluation Via Subcircuits}
	The main idea in the proof of the following lemma is to compute the evaluation values subcircuit-wise, starting with the deepest ones, and then iteratively moving up the circuit in $\depth(C_n)/s$ rounds, passing on the newly computed values to the right reducers, until the value of the unique output node is known. 
	
	\begin{lemma}\label{updowncircuits}
		If all up-circuits and down-circuits are constructed in the proper reducers, $C_n$ can be evaluated in $O(\depth(C_n)/\log n)$ $\DMRC$-rounds.
	\end{lemma}

	\begin{proof}
		Without loss of generality, let $\depth(C_n)$ be divisible by $\subcircuitDepth$ and let $r=\depth(C_n)/s$. Once all $L_i$-down-circuits and $L_i$-up-circuits for all $i\in\{1,\ldots,r\}$ have been constructed, we can evaluate $C_n$ on the given input assignment. We begin by evaluating the $L_{r-1}$-down-circuits. Since every input node has its value assigned in a $v$-down-circuit, the $L_{r-1}$-down-circuits can be computed in the reducers containing these $v$-down-circuits. With the values of all nodes at level $L_{r-1}$ determined, we can send the necessary values to the $L_{r-2}$-down-circuits and, in the case of edges that were divided using two dummy nodes, to lower-level down-circuits. 
		Nodes at level $L_{r-1}$ that are necessary to compute $L_{r-2}$-down-circuits are described in the $L_{r-1}$-up-circuits.
		Any node $v$ at level $L_{r-1}$ that is necessary to compute $L_{r-2}$-down-circuits is described in the $v$-up-circuit. 
		Therefore, the output of the reducer$_q$ is as follows: 
		Let $v$ be at level $L_{r-1}$ and let $u_i$, for $i\in\{1,\ldots,k_v\}$, be the 
		nodes at level $L_{r-2}$ in the $v$-up-circuit.
		For each $v$ in reducer$_q$, it outputs $(\To(u_i),v,\val(v))$, where $\To(u_i)$ is the index of the reducer containing the $u_i$-down-circuit and $\val(v)$ is the value of $v$ determined in the computation of $v$-down-circuit.
		The reducer$_q$ also passes on all $v$-down-circuits and $v$-up-circuits contained in it. %
		
		In the next round, the map function 
		sends each 
		$(\To(u_{\In}),v,\val(v))$ 
		to the reducer containing the $u_{\In}$-down-circuit;
		that is, 
		it generates the key-value pair $\langle \To(u_{\In}); (v, \val(v))\rangle$. 
		Of course, the map function also passes along all $v$-down-circuits and $v$-up-circuits to the proper reducers.
		
		Since now each $L_{r-2}$-down-circuit is contained completely in a reducer that has gathered all values of nodes at level $L_{r-1}$ necessary to compute this subcircuit, all $L_{r-2}$-down-circuits can be computed in their reducers. 
		Now we can compute the values of nodes higher and higher up in the circuit, by iterating the last mapping-reducing function pair, until the value is finally known for the unique output node. 
		
		As before, we clearly stay within the memory and I/O buffer limits of each reducer. 
	\end{proof}

	\section{Conclusion and Research Opportunities}\label{Sec4}
	In a substantial improvement over all previously known results, we have shown that $\NC^{i+1} \subseteq \DMRC^{i}$ for all $i\in\N$. In the case of $\NC^1 \subseteq \DMRC^0$, we have proved this result for every feasible choice of $\epsilon$ in the model, that is, for $0<\epsilon\le1/2$. For $i>0$, we have shown the result to hold for all but one value, namely $\epsilon=1/2$.
	
	Achieving these two results required a detailed description of two different, delicate simulations within the MapReduce framework. For the case of $\NC^1$, which is particularly relevant in practice, we applied Barrington's theorem and simulated width-bounded branching programs~\cite{B}, whereas we directly simulated the circuits for the higher levels of the hierarchy. We emphasize that none of the two approaches can replace the other: Barrington's theorem only gives a characterization for the first level of the $\NC$ hierarchy and the second approach does not even yield $\NC^1\subseteq \MRC^0$. 
	(Recall that $\DMRC$ is just the deterministic variant of $\MRC$, so we have $\DMRC^i\subseteq\MRC^i$ for all $i\in\N$.)
	
	We would like to briefly address the small question that immediately arises from our result, namely whether it is possible to extend the inclusion $\NC^{i+1} \subseteq \DMRC^i$ of Theorem~\ref{ACiMRCi} to the case $\epsilon=1/2$. Going through all involved lemmas, we see that the two reasons that our proof does not work in this corner case are the sorting of the nodes using  Lemma~\ref{Sort-Kar} and the construction of the up-circuits and down-circuits in Lemma~\ref{computing-down-up-circuits}. Regarding the former, we can avoid the restriction by allowing randomization. For the latter, it is not clear that this can be achieved, however. If there was any way to construct the levels for $\epsilon=1/2$ as well, then Theorem~\ref{ACiMRCi} would immediately extend to the full range $0<\epsilon\le 1/2$ of feasible choices for $\epsilon$. 
	
	Besides dealing with the small issue mentioned above, the natural next step for future research is to take the complementary approach and address the reverse relationship: Having shown in this paper how to obtain efficient deterministic MapReduce algorithms for parallelizable problems, we now aim to include the largest possible subset of $\DMRC^i$ into $\NC^{i+1}$ for all $i\in\N$. 
	By simply padding a $\mathcal{P}$-complete language so as to include it in $\DMRC^0$, Karloff et al.~\cite[Thm. 4.1]{KSV} proved that $\DMRC^0\subseteq \NC$ would imply $\mathcal{P}=\NC$, an equality generally deemed unlikely to hold. 
	Thus we cannot expect to prove $\DMRC^i=\NC^{i+1}$, but try to determine $\DMRC^i\cap\NC^{i+1}$ in order to finally settle the long-standing open question of how exactly the MapReduce classes correspond to the classical classes of parallel computation. 

	\bibliographystyle{plainurl}
	
\end{document}